\documentclass[12pt]{article}

\usepackage{amsmath,amsfonts,amsthm,amssymb}
\usepackage{setspace}
\usepackage{fancyhdr}
\usepackage{lastpage}
\usepackage{extramarks}

\usepackage[moderate]{savetrees}

\usepackage{chngpage}
\usepackage{soul,color}
\usepackage{graphicx,float,wrapfig}
\usepackage{listings}
\usepackage{tikz}
\usepackage{harpoon}
\usepackage[ruled]{algorithm2e}
\usepackage{hyperref}
\usepackage{bm}
\usepackage{indentfirst}
\usepackage{natbib}
\usepackage{verbatim}
\usepackage{pgfplots}
\pgfplotsset{compat=1.18}

\hypersetup{
    colorlinks=true,
    citecolor=blue,
    linkcolor=blue,
    urlcolor=blue,
    bookmarksopen=true,
    pdfstartview={XYZ null null 1.00},
    pdfpagelayout={SinglePage}
}

\newtheorem{theorem}{Theorem}
\newtheorem{lemma}{Lemma}
\newtheorem{proposition}{Proposition}
\newtheorem{corollary}{Corollary}
\newtheorem{definition}{Definition}

\usetikzlibrary{arrows,graphs}
\newcommand{\enabstractname}{Abstract}

\newenvironment{enabstract}{%
  \noindent\mbox{}\hfill{\bfseries \enabstractname}\hfill\mbox{}\par
  \vskip 2.5ex}{\par\vskip 2.5ex}

\graphicspath{{figure/}}                                 
\usepackage[a4paper]{geometry}                           
\begin{document}
\title{On the Complexity of Bayesian Signal Processing}
\date{August 31, 2026}
\author{Yi Liu\thanks{\baselineskip=0.7\normalbaselineskip
	Department of Economics, Yale University. {\tt yi.liu.yl2859@yale.edu}}}
\maketitle
\begin{enabstract}
We develop a computational framework for Bayesian decision-making. We show that as long as no action is optimal in every state, Bayes-optimal choice is intractable. This hardness need not arise from large action, state, or signal spaces, nor from a complicated represented utility function: extracting enough information from a hard-to-interpret signal to act optimally can itself be computationally hard. We also characterize tractability across approximation notions and identify their sources of difficulty. Under the probably approximately correct criterion, sample-based Bayesian learning is tractable if and only if the signal support is bounded. Our results provide justifications for bounded rationality, costly Bayesian inference, and sample-based Bayesian learning.

\vspace{1ex}
\noindent \textbf{Keywords:} bounded rationality, Bayesian decision theory, computational complexity, Bayesian learning
\end{enabstract}

\newpage
\section{Introduction}

Inference through Bayes' rule is a standard assumption in Bayesian decision theory. This assumption provides a foundational motivation for the Blackwell-experiment framework, under which every signal-generating process can, without loss of generality, be represented in likelihood-table form. However, implementing this assumption requires decision makers to have infinite computational power. This leaves a natural question: Can a computationally bounded decision maker act optimally, as Bayesian decision theory predicts?

In real-world signal-processing problems, the decision maker typically knows how signals are generated but is not given the corresponding likelihood table. An investor may understand how a firm's underlying condition, accounting records, managerial reporting choices, and auditing procedure produce its final disclosure; a regulator may understand how a drug's underlying safety profile, trial design, patient outcomes, and laboratory procedures produce a safety-test result; and a user may learn from a platform's disclosure how profile data, item attributes, and its algorithm produce a recommendation. In each case, the decision maker observes the realized output and understands how it is produced without being given its state-contingent likelihoods. Therefore, signal processing is essentially the computational problem of transforming knowledge of the signal-generating process and the realized signal into a decision. Knowing how a signal is generated does not itself provide the likelihood table required for Bayesian choice.

We analyze signal processing by reframing Bayesian decision theory within a computational framework. To isolate the difficulty arising from signal processing, we fix the decision environment: the state space, action space, and utility function are all fixed, with utility determined by the unobserved state and the chosen action. The signal-generating process is represented by a randomized sampler circuit that takes the state and a random tape as inputs and produces a binary-string signal. The decision maker's choice rule is modeled as a selector algorithm that takes the description of the sampler circuit and the realized signal as inputs and returns an action as output. We adopt the classical view in computer science that a selector is tractable if and only if it always runs in time polynomial in the size of its input. Our analysis covers selectors that always output a Bayes-optimal choice as well as selectors that output an approximately optimal choice. Our goal is to determine which notions of optimality--exact Bayes optimality or approximate optimality--admit a tractable selector and to identify where the hardness comes from. 

Our first result (\autoref{thm:exact-bayes-dichotomy} and \autoref{cor:hard-exact-small-signal}) establishes a sharp dichotomy for exact Bayes-optimal choice. Under a standard complexity assumption, Bayes-optimal choice is tractable for every signal-generating process if and only if one action is optimal in every state. The difficulty comes from two requirements: a posterior can lie arbitrarily close to an action boundary, creating an extreme precision requirement for comparing implicitly represented likelihoods, and the selector must succeed after every supported signal. This hardness already appears with only two supported signals, so it need not be driven by a large signal space.

\autoref{thm:signalwise-approximation-threshold} and \autoref{cor:hard-signalwise-small-signal} further show that relaxing exactness does not gradually restore tractability when performance must still be guaranteed after every supported signal. There is a sharp threshold: an efficient deterministic selector exists only for guarantees that can already be attained by ignoring the signal and always choosing the same action. Any improvement over this signal-independent benchmark, however small, is intractable under a standard complexity assumption. Since this result also holds with only two supported signals, the universal quantifier--the requirement that the selector succeed after \emph{every} supported signal--is itself a source of difficulty, and the hardness is not necessarily caused by the size of the signal space.

In \autoref{thm:ex-ante-approximation-threshold} and \autoref{thm:pac-ex-ante-threshold}, we relax the guarantee further by requiring good approximate performance only over a high-probability mass of signals, without requiring the guarantee to hold for every supported signal. Even so, under a standard complexity assumption, neither an efficient deterministic selector nor an efficient randomized probably approximately correct  (PAC) selector can uniformly improve on the no-information benchmark over unrestricted signal-generating processes. However, \autoref{prop:finite-support-pac-upper} implies that the hard instances rely \emph{solely} on possibly exponentially large signal support. With a known polynomial bound on support, a Monte Carlo empirical-Bayes selector efficiently achieves any fixed nonexact approximation with any fixed confidence level (\autoref{prop:finite-support-pac-upper}). Thus, for PAC ex ante performance, the potentially exponential size of the signal support is the only remaining source of computational difficulty. 

In summary, we show that Bayesian signal processing is essentially a computationally hard problem. We identify the universal quantifier as the main source of this hardness: it requires good performance after every signal, including signals that may be difficult to interpret. The resulting difficulty includes, but is not limited to, the burden created by a potentially large signal space. We also provide a justification for using sample-based Bayesian learning in economics, computer science, and statistics. This approach appears to be the way to solve the weaker approximation problem when the signal support is bounded--the only problem in our framework that may admit a tractable solution.

\paragraph{Related Literature.}

Our work relates to a broad literature in economics on bounded rationality and the computational limits of choice and information processing. \citet{Camara2022} studies rational and computationally tractable choice over high-dimensional product menus, where hardness comes from the large, input-dependent action space and additive separability restores tractability. By contrast, here the finite action space and utility matrix are fixed; the difficulty lies instead in exact inference from a succinct experiment and possibly large signal space. \citet{Wilson2014} models bounded rationality in Bayesian information processing by imposing a finite-memory constraint. By contrast, we impose no exogenous memory bound; instead, we ask whether the requisite inference from a sampler representation can be carried out in an asymptotic time limit. \citet{HazlaEtAl2021} establish hardness of Bayesian decisions in social-learning networks under fixed natural utilities. Their hardness comes from the exponentially large joint private signal space created by neighbors in the network. \citet{EcheniqueGolovinWierman2011} find that the hardness of the consumer's problem can arise from a complex representation of the utility function. By contrast, our hardness result admits a clear economic interpretation and is independent of how utility is revealed. \citet{Eboli2003} studies how the computational cost of Bayesian updating depends on the representation of information. \citet{Grether1980} finds experimental evidence of base-rate neglect, especially among inexperienced or weakly incentivized subjects. 

Our work also relates to research in machine learning and statistics on learning Bayes-optimal decision rules. \citet{Haussler1992} develops a decision-theoretic PAC-learning framework for learning bounded-loss decision rules from labeled samples when the data distribution is unknown; \citet{Elkan2001} and \citet{AvilaPiresEtAl2013} study cost-sensitive classification, while \citet{AudibertTsybakov2007} analyzes plug-in classifiers that estimate conditional class probabilities before applying the Bayes rule.  \citet{LeEtAl2017} amortize posterior inference, \citet{SainsburyDaleEtAl2024} train neural Bayes estimators, and \citet{PudloEtAl2016} formulate approximate Bayesian computation model choice as a classification problem, training a random forest on model-labeled data generated from the prior predictive distribution.  \citet{OrtizKaelbling2000} use sampling to select near-optimal actions in one-decision influence diagrams, while \citet{AlsingEtAl2023}, \citet{GoreckiEtAl2024}, and \citet{PolsonEtAl2024} learn expected-utility surfaces, posterior expected costs, or posterior distributional utilities that support action choice.  These approaches are generally model-specific and sample-based, emphasizing statistical risk, amortization, or simulation efficiency. Thus, these studies develop model-specific, simulation-based methods for directly learning Bayes actions, which inspires our construction of the empirical Bayes selector. By contrast, we treat a sampler description as part of the input and ask whether one polynomial-time selector can implement or approximate the Bayes optimal choice uniformly over all such samplers in a fixed finite decision environment.  Our contribution is the resulting fixed-environment classification of exact and approximate tractability, together with finite-support sample bounds under an explicit sample-only access model.

\section{Model}

A computationally bounded decision maker faces a finite decision making problem where the state $\theta\in\Theta$ is unknown to the decision maker who will choose an action $a\in A$. The state is drawn according to a full-support prior $\mu_0\in\Delta(\Theta)$ where every component is rational, and the decision maker's utility function is specified by
\[
u:A\times\Theta\longrightarrow\mathbb{Q}_{>0}.
\]
These objects constitute the decision environment and are fixed.  

Before choosing an action, the decision maker receives a signal from a sampler-described experiment. A sampler $q$ is represented by a finite randomized Boolean circuit that, given input $\theta$ and a uniformly random seed $R$, outputs a signal $s = q(\theta;R)$. Therefore, it induces the Blackwell experiment 
\[
q(s|\theta) = \Pr_{R}(q(\theta;R) = s).
\]
By a slight abuse of notation, we use \(q\) for both the sampler circuit and the experiment it induces. Fix a standard gate-list encoding of sampler descriptions and write $\langle q\rangle$ for the bit string encoding $q$. Appendix~\ref{appendix:computational-preliminaries} provides an informal guide to circuits, gate-list encodings, and the complexity classes used below. For any binary string $x$, let $|x|$ denote its number of bits.  Thus $|\langle q\rangle|$ is the description length of the sampler, whereas $|s|$ is the length of the realized signal string.  The size of an instance is
\[
|\langle q\rangle|+|s|.
\]

We model the decision maker's choice rule as a \emph{selector} $\alpha(\langle q\rangle,s)$. A selector is a function that maps $(\langle q\rangle,s)$ to an action $\alpha(\langle q\rangle,s)\in A$. Because the decision environment $(\Theta,A,\mu_0,u)$ is fixed, the sampler description $\langle q\rangle$ and the realized signal $s$ are the selector's only varying arguments.

For a signal $s$, let
\[
Z(q,s)=\sum_{\theta\in\Theta}\mu_0(\theta)q(s\mid\theta).
\]
Posterior statements below are restricted to $Z(q,s)>0$.  Define the unnormalized posterior payoff
\[
W_a(q,s)=\sum_{\theta\in\Theta}
\mu_0(\theta)q(s\mid\theta)u(a,\theta).
\]
When $Z(q,s)>0$, the posterior expected payoff is $W_a(q,s)/Z(q,s)$, so its maximizers are exactly the maximizers of $W_a(q,s)$.

We call a selector $\alpha$ an \emph{exact Bayes selector} if it satisfies
\[
\alpha(\langle q\rangle,s)\in\arg\max_{a\in A}W_a(q,s)
\]
on every input. 

A selector $\alpha$ is a \emph{deterministic polynomial-time selector} if there exists a  deterministic algorithm that, on every input $(\langle q\rangle,s)$, outputs $\alpha(\langle q\rangle,s)$ in time polynomial in $|\langle q\rangle|+|s|$. We call such a selector \emph{tractable}. Unless explicitly stated as randomized, we also use \emph{polynomial-time selector} to refer to this deterministic notion.

\section{Exact Bayesian Selectors}

Define the set of dominant actions
\[
D(u)=\left\{d\in A:
u(d,\theta)\geq u(a,\theta)
\text{ for every }a\in A,\ \theta\in\Theta\right\}.
\]

\begin{theorem}
\label{thm:exact-bayes-dichotomy}
If $\mathsf{P}\neq\mathsf{PP}$, there is a polynomial-time exact Bayes selector for all sampler descriptions $q$ and signals $s$ if and only if $D(u)\neq\emptyset$.
\end{theorem}

\begin{proof}[Proof Overview]
If $D(u)\neq\emptyset$, the claim is immediate: a selector can always return a dominant action, independently of the sampler and the observed signal, and hence runs in constant time. When $D(u)=\emptyset$, we reduce from \textsc{Majsat} (Majority-SAT), a canonical $\mathsf{PP}$-complete problem \citep{Gill1977}. The reduction uses a sampler with only two possible signals and translates the acceptance probability of a Boolean circuit into the location of a posterior along a line segment. We sketch the main idea here and leave the technical details to Appendix~\ref{appendix:proof}.

Because no action is dominant, one can find two nearby posterior beliefs $p^-$ and $p^+$ on opposite sides of a decision boundary. Let $p^0=(p^-+p^+)/2$. Along the segment joining $p^-$ and $p^+$, the uniquely optimal payoff vector is that of some action $a^-$ on the $p^-$-side of $p^0$, and that of another action $a^+$ on the $p^+$-side. Thus, an exact selector evaluated at a posterior on this segment reveals on which side of $p^0$ that posterior lies.

Now take a \textsc{Majsat} instance $\varphi$. We can determine whether the acceptance probability of $\varphi$ is greater than $1/2$ or not. For each state $\theta$, let $E_\theta^-$ and $E_\theta^+$ be auxiliary random events calibrated so that they implement the posteriors $p^-$ and $p^+$, respectively, upon observing a target signal $s^\star$. Using independent random seeds $R$ and $U$ for $\varphi$ and these auxiliary events, define the two-signal sampler
\[
q_\varphi(\theta;R,U)=
\begin{cases}
s^\star,
& \bigl(\varphi(R)=0\ \wedge\ E_\theta^-(U)\bigr)
  \ \vee\ \bigl(\varphi(R)=1\ \wedge\ E_\theta^+(U)\bigr),\\
s^\circ, & \text{otherwise}.
\end{cases}
\]
By Bayes' rule,
\[
\Pr(\cdot\mid s^\star)
=(1-z_\varphi)p^-+z_\varphi p^+.
\]
Therefore, the posterior lies on the $p^+$-side of $p^0$ exactly for yes-instances of \textsc{Majsat}, and on the $p^-$-side otherwise. A fixed decoder applied to the output of any exact Bayes selector consequently determines whether $\varphi\in\textsc{Majsat}$. 
\end{proof}

\autoref{thm:exact-bayes-dichotomy} reveals two features that jointly generate this computational difficulty. First, exact optimality can require extreme precision. Whenever acting optimally requires some information about the state--that is, whenever no single action is optimal in every state--an experiment can place the posterior arbitrarily close to a decision boundary, so that even a minute difference in signal likelihoods changes the optimal action. Second, the guarantee is uniform over signals: an exact Bayes selector must choose optimally after \emph{every} supported signal and hence cannot ignore even a single signal whose interpretation is computationally difficult. Thus, even though the underlying decision problem has only finitely many states and actions, exact Bayesian choice may require \emph{arbitrarily fine} likelihood distinctions to be resolved after \emph{every} possible observation. To further explore the driven force of the hardness, the same argument of the proof yields the following corollary.
\begin{corollary}\label{cor:hard-exact-small-signal}
Suppose $D(u)=\emptyset$ and $\mathsf{P}\neq\mathsf{PP}$. There is no polynomial-time exact Bayes selector even for the restricted class of samplers $q$ satisfying
\[
\left|\{s:Z(q,s)>0\}\right|\leq 2.
\]
\end{corollary}
\autoref{cor:hard-exact-small-signal} separates the role of the requirement to succeed after \emph{every} supported signal from the size of the signal space. The hardness persists even when the experiment has at most two supported signals, so the cost of this universal quantifier cannot be solely attributed to having to cover an exponentially large signal space. Its force is instead that the selector is not permitted to fail on even the one signal that encodes a difficult likelihood comparison. This distinction raises the natural question studied in the next section: does the hardness disappear if exactness, and hence the required precision, is relaxed while the every-signal requirement is retained? \autoref{thm:signalwise-approximation-threshold} and \autoref{cor:hard-signalwise-small-signal} show that it does not. 

\section{Approximate Selectors}
\label{sec:deterministic-approximation}
In this section, we consider two notions of multiplicative approximation and identify the source of complexity associated with each. 

\subsection{Signalwise approximation}
\label{subsec:signalwise-approximation}

To relax exactness while retaining the requirement that the guarantee hold after every supported signal, we introduce the following notion of signalwise approximation. 

\begin{definition}[Signalwise deterministic approximation]
\label{def:signalwise-approximation}
A deterministic signalwise $c$-selector is a selector $\alpha$ such that
\begin{equation}
\label{eq:signalwise-guarantee}
W_{\alpha(\langle q\rangle,s)}(q,s)
\geq
c\max_{a\in A}W_a(q,s)
\end{equation}
for every sampler $q$ and every signal $s$ satisfying $Z(q,s)>0$.
\end{definition}

For each state, let $M_\theta:=\max_{a\in A}u(a,\theta)$, and define
\begin{equation}
\label{eq:gamma-sig}
\gamma_{\mathrm{sig}}(u)
:=
\max_{d\in A}\min_{\theta\in\Theta}
\frac{u(d,\theta)}{M_\theta}.
\end{equation}

\begin{theorem}
\label{thm:signalwise-approximation-threshold}
Fix $c\in(0,1]$.  If $\mathsf{P}\neq\mathsf{NP}$, then there is a deterministic polynomial-time signalwise $c$-selector if and only if $c\leq\gamma_{\mathrm{sig}}(u)$.
\end{theorem}

\begin{proof}[Proof Overview]
When $c\leq\gamma_{\mathrm{sig}}(u)$, the constant selector that always returns an action attaining the maximum in \eqref{eq:gamma-sig} satisfies the approximation requirement. For the hardness direction, fix $c>\gamma_{\mathrm{sig}}(u)$ and suppose that a polynomial-time signalwise $c$-selector $\alpha$ exists. We show how to use $\alpha$ to solve \textsc{SAT}, a standard $\mathsf{NP}$-complete problem.

The key observation is that, if $\alpha$ outputs an action $a$ after observing a signal $s$, then $s$ cannot perfectly reveal a state
\[
\theta_i\in\arg\min_{\theta\in\Theta}
\frac{u(a,\theta)}{M_\theta}.
\]
Indeed, if $s$ perfectly revealed $\theta_i$, then
\[
\frac{W_a(q,s)}{\max_{a'\in A}W_{a'}(q,s)}
=\frac{u(a,\theta_i)}{M_{\theta_i}}
\leq\gamma_{\mathrm{sig}}(u)<c,
\]
contradicting the signalwise $c$-approximation guarantee.

To exploit this observation, let $\Theta=\{\theta_1,\ldots,\theta_m\}$ and, for $m$ formulas $\psi_1,\ldots,\psi_m$, construct the two-signal sampler
\[
q(\theta;R) = \begin{cases}
s^\star,   & \theta=\theta_i\text{ for some }i\in\{1,\ldots,m\}
    \text{ and }\psi_i(R)=1,\\
s^\circ, & \text{otherwise}.
\end{cases}
\]
After calling $\alpha(\langle q\rangle,s^\star)$ and obtaining $a$, choose an index $i$ minimizing $u(a,\theta_i)/M_{\theta_i}$. The formula $\psi_i$ cannot be the unique satisfiable formula; otherwise, $s^\star$ would perfectly reveal $\theta_i$. Hence,
\[
\bigvee_{j=1}^m\psi_j\text{ is satisfiable}
\quad\Longleftrightarrow\quad
\bigvee_{j\neq i}\psi_j\text{ is satisfiable}.
\]

Repeatedly applying this deletion rule while pruning the standard \textsc{SAT} self-reduction tree keeps the list of residual formulas at constant size after each branching step. Once all variables have been assigned, the remaining assignments can be checked directly, yielding a polynomial-time algorithm for \textsc{SAT}, a contradiction.
\end{proof}

\autoref{thm:signalwise-approximation-threshold} answers the question raised in the previous section: exactness is not essential to the hardness. The threshold $\gamma_{\mathrm{sig}}(u)$ has a direct information-processing interpretation. When $c\leq\gamma_{\mathrm{sig}}(u)$, a single fixed action attains the guarantee without consulting either the sampler or the realized signal, so information processing is dispensable. Once $c>\gamma_{\mathrm{sig}}(u)$, however, the guarantee requires the selector to extract some decision-relevant information from its input, and even an arbitrarily small improvement over the signal-independent benchmark is intractable under $\mathsf{P}\neq\mathsf{NP}$. Thus, when performance must be guaranteed after every supported signal, approximation does not gradually restore tractability: tractability returns only when the requirement is weak enough that the signal can be ignored altogether. Moreover, the hardness construction uses only two signals, yielding the following corollary.
\begin{corollary}\label{cor:hard-signalwise-small-signal}
Fix $c\in(\gamma_{\mathrm{sig}}(u),1]$ and assume $\mathsf{P}\neq\mathsf{NP}$. There is no deterministic polynomial-time signalwise $c$-selector even for the restricted class of samplers $q$ satisfying
\[
\left|\{s:Z(q,s)>0\}\right|\leq 2.
\]
\end{corollary}

\autoref{cor:hard-signalwise-small-signal} shows that the hardness need not be driven by the straightforward difficulty created by a large signal space. An exponentially large support can certainly impose a computational burden on procedures that must enumerate or cover many possible signals, but such a large support is not necessary here: the hardness persists even when the experiment has at most two supported signals. The force of the ``every-signal'' requirement therefore lies in its quantifier, not in the number of signals. A signalwise selector is not permitted to average away, ignore, or sacrifice even the one signal whose state-contingent likelihoods encode a hard computational distinction. Together, \autoref{thm:signalwise-approximation-threshold} and \autoref{cor:hard-signalwise-small-signal} show that, as soon as the desired guarantee exceeds the best signal-independent benchmark, however slightly, retaining this universal requirement can force the selector to resolve an intractable likelihood comparison at a signal it is not allowed to ignore. Thus, beyond any difficulty caused by the size of the signal space, the central information-processing barrier highlighted here is the requirement to extract enough decision-relevant information from every supported signal.

This conclusion raises a natural question: does the hardness persist if good approximate performance is required only over a high-probability mass of signals, rather than after every supported signal? \autoref{thm:ex-ante-approximation-threshold} and \autoref{thm:pac-ex-ante-threshold} show that it does, while \autoref{prop:finite-support-pac-upper} indicates that the remaining difficulty comes only from having to explore a potentially large signal space.

\subsection{Ex ante approximation}
\label{subsec:ex-ante-approximation}

To formalize good aggregate performance over a high-probability mass of signals, we introduce the following notion of ex ante approximation. For a deterministic selector $\alpha$, define its ex ante payoff by
\begin{equation}
\label{eq:selector-ex-ante-value}
\operatorname{VAL}_\alpha(q)
:=
\sum_s W_{\alpha(\langle q\rangle,s)}(q,s),
\end{equation}
and define the unrestricted Bayes-optimal ex ante payoff by
\[
\operatorname{OPT}(q)
:=
\sum_s\max_{a\in A}W_a(q,s).
\]

\begin{definition}[Ex ante approximation]\label{def:ex_ante_approximation}
A deterministic ex ante $c$-selector is a selector $\alpha$ such that
\[
\operatorname{VAL}_\alpha(q)
\geq c\,\operatorname{OPT}(q)
\]
for every sampler $q$.
\end{definition}

The relevant full-information and no-information benchmarks are
\[
V^{\mathrm{FI}}
:=
\sum_{\theta\in\Theta}\mu_0(\theta)M_\theta 
\qquad \text{and} \qquad
V^{\mathrm{NI}}
:=
\max_{a\in A}
\sum_{\theta\in\Theta}\mu_0(\theta)u(a,\theta).
\]
Define
\[
\gamma_{\mathrm{EA}}(u,\mu_0)
:=
\frac{V^{\mathrm{NI}}}{V^{\mathrm{FI}}}.
\]

\begin{theorem}
\label{thm:ex-ante-approximation-threshold}
Fix $c\in(0,1]$. If $\mathsf{BPP}\neq\mathsf{SZK}$, there is a deterministic polynomial-time ex ante $c$-selector if and only if
\begin{equation}
\label{eq:ex-ante-threshold}
c\leq\gamma_{\mathrm{EA}}(u,\mu_0).
\end{equation}
\end{theorem}

The assumption $\mathsf{BPP}\neq\mathsf{SZK}$ is stronger than $\mathsf{P}\neq\mathsf{NP}$.\footnote{Indeed, $\mathsf{BPP}\neq\mathsf{SZK}$ implies $\mathsf{P}\neq\mathsf{NP}$; see \citet{SahaiVadhan2003}.} Nevertheless, the assumption is widely believed by computer scientists: several cryptosystems rely on the presumed hardness of problems in $\mathsf{SZK}$, so an efficient algorithm for all of $\mathsf{SZK}$ would compromise these systems. There is also an evidence for the assumption: \citet{bouland2019power} construct an oracle $A$ relative to which $\mathsf{BPP}^{A}\subsetneq \mathsf{SZK}^{A}$.

\begin{proof}[Proof Overview]
If $c\leq\gamma_{\mathrm{EA}}(u,\mu_0)$, the constant selector that always returns a prior-optimal action satisfies the ex ante $c$-approximation guarantee: it achieves $V^{\mathrm{NI}}$, while the Bayes-optimal value is at most $V^{\mathrm{FI}}$. For the hardness direction, fix $c>\gamma_{\mathrm{EA}}(u,\mu_0)$ and suppose that a deterministic polynomial-time ex ante $c$-selector $\alpha$ exists. We use $\alpha$ to construct a randomized polynomial-time algorithm for Statistical Difference, an $\mathsf{SZK}$-complete problem \citep{SahaiVadhan2003}, with two-sided error probability at most $1/3$.

After applying the standard polarization argument, it suffices, for a sufficiently small $\varepsilon>0$, to distinguish between
\[
d_{\mathrm{TV}}(P_0,P_1)\leq\varepsilon
\qquad\text{and}\qquad
d_{\mathrm{TV}}(P_0,P_1)\geq1-\varepsilon,
\]
where $d_{\mathrm{TV}}$ denotes total variation distance. Assign each state $\theta$ a distinct binary codeword. We construct a sampler $q$ by replacing each bit of the codeword with an independent draw from $P_0$ or $P_1$, according to whether that bit is zero or one.

If $d_{\mathrm{TV}}(P_0,P_1)\geq1-\varepsilon$, the samples almost reveal every bit of the codeword and hence almost reveal the state. Therefore,
\[
\operatorname{OPT}(q)\geq V^{\mathrm{FI}}-O(\varepsilon),
\qquad
\operatorname{VAL}_\alpha(q)
\geq c\bigl(V^{\mathrm{FI}}-O(\varepsilon)\bigr).
\]
If $d_{\mathrm{TV}}(P_0,P_1)\leq\varepsilon$, the signal distributions are nearly the same across states, so the selector can obtain at most the no-information value up to a small error:
\[
\operatorname{VAL}_\alpha(q)
\leq V^{\mathrm{NI}}+O(\varepsilon).
\]
Because $c>\gamma_{\mathrm{EA}}(u,\mu_0)=V^{\mathrm{NI}}/V^{\mathrm{FI}}$, these two bounds are separated by a fixed positive gap when $\varepsilon$ is sufficiently small. Therefore, we can estimate $\operatorname{VAL}_\alpha(q)$ by sampling and use this gap to distinguish the two cases with bounded error, placing Statistical Difference in $\mathsf{BPP}$ and contradicting $\mathsf{BPP}\neq\mathsf{SZK}$.
\end{proof}

\autoref{thm:ex-ante-approximation-threshold} shows that the hardness survives the move from signalwise to ex ante approximation. Relaxing the guarantee does not eliminate the information-processing problem: once the target exceeds the no-information benchmark, the decision maker must condition her action on the observed signal, and even achieving good aggregate performance over a high-probability mass of signals remains intractable. Unlike the earlier every-signal hardness results, however, this difficulty may have a more straightforward source. A high-probability mass can itself be dispersed over exponentially many distinct signals, so performing well on that mass may require the selector to explore and learn the decision implications of exponentially many realizations. In fact, the reduction in the proof of \autoref{thm:ex-ante-approximation-threshold} uses an exponentially large signal support: each realized signal has polynomial length so that there are exponentially many possible realizations generated from the outputs of $P_0$ and $P_1$. Therefore there is open problem  whether deterministic ex ante hardness persists under polynomially bounded support; we currently have neither a polynomial-time selector nor a matching hardness result for this restricted case. We believe that the problem becomes tractable once restricted to signal spaces of polynomially bounded size.

The next section provides supporting evidence by considering a slight relaxation: the ex ante approximation guarantee is required to hold only with high probability, rather than with probability one, thereby allowing randomized selectors. Under this relaxation, the hardness stems solely from the large signal space.

\section{Probably Approximately Correct Selectors}
\label{sec:pac-selectors}

In this section, we turn to randomized selectors and require their approximation guarantees to hold only with high probability.

\subsection{PAC ex ante selectors}
\label{subsec:pac-ex-ante-selectors}

We begin by formally defining randomized selectors and the corresponding notion of probably approximately correct (PAC) ex ante approximation.

\begin{definition}[PAC ex ante selector]
\label{def:pac-ex-ante-selector}
A randomized polynomial-time selector is a uniform probabilistic algorithm $\alpha$ that, on input $(\langle q\rangle,s)$ and random tape $R$, outputs
$\alpha(\langle q\rangle,s;R)\in A$ in time polynomial in $|\langle q\rangle|+|s|$ for every $R$.  For fixed $R$, write $\alpha_R(\langle q\rangle,s):=\alpha(\langle q\rangle,s;R)$.

For fixed parameters $c\in(0,1]$ and $\delta\in(0,1)$, $\alpha$ is a \emph{$(c,\delta)$-PAC ex ante selector} if, for every sampler $q$,
\[
\Pr_R\!\left[
\operatorname{VAL}_{\alpha_R}(q)
\geq c\,\operatorname{OPT}(q)
\right]
\geq1-\delta.
\]
Here $\operatorname{VAL}$ is defined in \eqref{eq:selector-ex-ante-value}.  The fair random tape $R$ is independent of the test signal and is held fixed throughout this ex ante value. $\alpha_R$ is a deterministic selector.
\end{definition}

The PAC notion also implies an expected-payoff guarantee for randomized selectors. Any $(c,\delta)$-PAC ex ante selector satisfies, for every sampler $q$,
\[
\mathbb E_R\!\left[\operatorname{VAL}_{\alpha_R}(q)\right]
\geq
c(1-\delta)\operatorname{OPT}(q).
\]
Indeed, the selector attains at least $c\operatorname{OPT}(q)$ with probability at least $1-\delta$; dropping its positive payoff on the remaining random tapes gives the inequality. Hence, a $(c,\delta)$-PAC guarantee yields a $c(1-\delta)$-approximation in expectation. The PAC analysis can therefore also be interpreted as an analysis of ex ante approximation for randomized selectors evaluated by their expected payoffs.

\begin{theorem}
\label{thm:pac-ex-ante-threshold}
Fix a factor $c\in(0,1]$ and an error probability $\delta\in(0,1)$. If $\mathsf{BPP}\neq\mathsf{SZK}$, there is a polynomial-time randomized $(c,\delta)$-PAC ex ante selector if and only if
\[
c\leq
\gamma_{\mathrm{EA}}(u,\mu_0).
\]
\end{theorem}
\begin{proof}[Proof Overview]
If $c\leq\gamma_{\mathrm{EA}}(u,\mu_0)$, the constant selector that always returns a prior-optimal action satisfies the $(c,\delta)$-PAC ex ante guarantee with probability one. For the hardness direction, fix $c>\gamma_{\mathrm{EA}}(u,\mu_0)$ and suppose that a polynomial-time randomized $(c,\delta)$-PAC ex ante selector $\alpha$ exists. We use the same polarized Statistical Difference instance and state-encoding sampler $q$ as in the proof of \autoref{thm:ex-ante-approximation-threshold}.

If $d_{\mathrm{TV}}(P_0,P_1)\geq1-\varepsilon$, then
\[
\operatorname{OPT}(q)\geq V^{\mathrm{FI}}-O(\varepsilon),
\]
so the PAC guarantee implies
\[
\Pr_R\!\left[
\operatorname{VAL}_{\alpha_R}(q)
\geq c\bigl(V^{\mathrm{FI}}-O(\varepsilon)\bigr)
\right]
\geq1-\delta.
\]
If $d_{\mathrm{TV}}(P_0,P_1)\leq\varepsilon$, the signal distributions are nearly the same across states. For every fixed random tape $R$, the resulting selector $\alpha_R$ is deterministic. By the similar argument,
 \[
\operatorname{VAL}_{\alpha_R}(q)
\leq V^{\mathrm{NI}}+O(\varepsilon)
\qquad\text{for every }R.
\]
Because $cV^{\mathrm{FI}}>V^{\mathrm{NI}}$, the two bounds are separated by a fixed positive gap for sufficiently small $\varepsilon$.

The only new issue is that the lower bound holds with probability $1-\delta$ over the selector's random tape, whereas the upper bound holds for every tape. We therefore draw a constant number of independent tapes. With high constant probability, at least one tape produces a high-value selector in the first case, while no tape can do so in the second. Estimating the value associated with each tape by sampling from $q$ and comparing the largest estimate with a threshold distinguishes the two cases with bounded error. This again places Statistical Difference in $\mathsf{BPP}$, contradicting $\mathsf{BPP}\neq\mathsf{SZK}$.
\end{proof}

\autoref{thm:pac-ex-ante-threshold} shows that allowing randomized selectors and requiring the approximation guarantee to hold only with high probability does not by itself alleviate the hardness identified in \autoref{subsec:ex-ante-approximation}. The next subsection shows that under this relaxation, the remaining hardness is driven entirely by the potentially exponential size of the signal space.

\subsection{Monte Carlo empirical Bayes selectors}
\label{subsec:monte-carlo-empirical-bayes}

In this section, we provide a construction of a selector that is based on Monte Carlo sampling from the sampler $q$. We also show its potentials to be computed in polynomial time. Write
\[
\mathcal S_q
:=
\bigcup_{\theta\in\Theta}\textbf{supp}\{q(\cdot\mid\theta)\},
\qquad
K_q:=|\mathcal S_q|,
\qquad
U:=\max_{a\in A,\,\theta\in\Theta}u(a,\theta),
\]
Fix a sample size $N$, a total order on $A$, and a fallback action $a^\circ\in A$.  Our Monte Carlo empirical Bayes selector has following three main steps: 
\begin{enumerate}
    \item \textbf{Draw and fix the training sample.}
    Because the selector is given $\langle q\rangle$, it can simulate $q(\cdot\mid\theta)$ in polynomial time.  Independently for each state $\theta$, draw
    \[
    S_{\theta,1},\ldots,S_{\theta,N}
    \sim q(\cdot\mid\theta).
    \]
    The resulting training sample is held fixed throughout evaluation. 

    \item \textbf{Form the empirical experiment.}
    For every state $\theta$ and signal $s$, define
    \[
    \widehat q_N(s\mid\theta)
    :=
    \frac{1}{N}\sum_{i=1}^N
    \mathbf 1\{S_{\theta,i}=s\}.
    \]
    The corresponding empirical unnormalized payoff of action $a$ at signal $s$ is
    \begin{equation}
    \label{eq:empirical-unnormalized-payoff}
    \widehat W_a(s)
    :=
    \sum_{\theta\in\Theta}
    \mu_0(\theta)\widehat q_N(s\mid\theta)u(a,\theta).
    \end{equation}

    \item \textbf{Define the fixed policy.}
    For every query signal $s\in\mathcal S_q$, set
    \[
    \widehat\alpha_N(s)
    \in
    \operatorname*{arg\,max}_{a\in A}\widehat W_a(s),
    \]
    with ties resolved by the fixed order, except that $\widehat\alpha_N(s)=a^\circ$ when all empirical counts at $s$ are zero. 
\end{enumerate}

The fixed training sample together with the evaluation rule is an implicit representation of $\widehat\alpha_N$.  On a query signal $s$, the selector can counts the observations equal to $s$ state by state, and computes the finitely many values in \eqref{eq:empirical-unnormalized-payoff}. Therefore, it need not enumerate $\mathcal S_q$ or store a complete signal-to-action table. 

We next analyze how many samples are sufficient for the empirical Bayes selector to satisfy a $(c,\delta)$-PAC guarantee.

\begin{proposition}
\label{prop:finite-support-pac-upper}
Assume $K_q=|\mathcal S_q|<\infty$.  Let $C_\mu := \sum_{\theta\in\Theta}\mu_0(\theta)^2.$ For every $c,\delta\in(0,1)$, if
\[
N
\geq
\frac{2U^2C_\mu}
{(1-c)^2(V^{\mathrm{NI}})^2}
\left(
K_q\log|A|
+
\log\frac{2}{\delta}
\right),
\]
then
\[
\Pr\!\left[
\operatorname{VAL}_{\widehat\alpha_N}(q)
\geq
c\,\operatorname{OPT}(q)
\right]
\geq1-\delta.
\]
\end{proposition}

Thus, because the decision environment is fixed, \autoref{prop:finite-support-pac-upper} implies that it suffices to draw
\[
N = O\!\left(
\frac{K_q+\log(1/\delta)}{(1-c)^2}
\right).
\]
This rate is tight, up to constants, within the class of fixed-training-sample forward-sampling selectors. In the fixed two-state matching environment and parameter range of \autoref{prop:monte-carlo-pac-lower}, where $K=K_q$, the lower bound proved in Appendix~\ref{appendix:monte-carlo-lower-bounds} matches the preceding upper bound. Hence the worst-case sample complexity in that setting is
\[
N = \Theta\!\left(
\frac{K+\log(1/\delta)}{(1-c)^2}
\right)
\]
samples per state.

More generally, suppose that a promise class comes with a known polynomial $p$ such that
\[
K_q\leq p(|\langle q\rangle|)
\]
for every sampler $q$ in the class. For every fixed pair $(c,\delta)$ with $c<1$ and $\delta\in(0,1)$, both the training phase and the evaluation of each signal then run in polynomial time. Thus, on this promise class, the Monte Carlo empirical Bayes selector is a uniform randomized polynomial-time $(c,\delta)$-PAC selector. 

Thus, \autoref{prop:finite-support-pac-upper} clarifies that, for fixed approximation and confidence parameters, the potentially exponential size of the signal support is the only remaining computational obstacle to achieving an ex ante approximation with high probability. The intuition is straightforward: when probability mass is spread across exponentially many signals, a tractable empirical Bayes selector can draw only polynomially many training samples, which may be insufficient to estimate the state-contingent signal distributions accurately enough for reliable decision making.

Moreover, a related point explains why a sampling-based method cannot satisfy the unconditional ex ante guarantee in \autoref{def:ex_ante_approximation}, even when the signal support has constant size. Any finite training sample has a positive--although possibly negligible--probability of being unrepresentative, in which case the empirical distribution and the induced policy may perform poorly. A negligible failure probability is still nonzero and therefore cannot satisfy a guarantee that must hold with probability one. Therefore, whether there is a tractable deterministic selector that can satisfy the ex ante approximation guarantee when the signal support has constant size still remains an open question.

\section{Discussions}

\subsection{Randomized selectors} 
Randomization affects exact and approximate signalwise selection differently. Suppose that a randomized selector is evaluated by its expected payoff conditional on every supported signal:
\[
\mathbb E_R\!\left[
W_{\alpha(\langle q\rangle,s;R)}(q,s)
\right]
\geq
c\max_{a\in A}W_a(q,s).
\]
For exact selection, randomization cannot help. Since no realized action can yield more than the Bayes-optimal payoff, attaining that payoff in expectation requires every action chosen with positive probability to be Bayes optimal. Randomization can therefore serve only as a tie-breaking device, and the exact-selection dichotomy in \autoref{thm:exact-bayes-dichotomy} remains unchanged.

For approximate selection, by contrast, randomization convexifies the action set and may improve the best signal-independent guarantee. For a lottery $\lambda\in\Delta(A)$, define the mixed-action benchmark
\[
\gamma_{\mathrm{sig}}^{\mathrm{mix}}(u)
:=
\max_{\lambda\in\Delta(A)}
\min_{\theta\in\Theta}
\frac{u(\lambda,\theta)}{M_\theta}.
\] 
Because pure actions are degenerate lotteries, $\gamma_{\mathrm{sig}}^{\mathrm{mix}}(u)\geq\gamma_{\mathrm{sig}}(u)$. For every $c<\gamma_{\mathrm{sig}}^{\mathrm{mix}}(u)$, a sufficiently accurate dyadic approximation to an optimizing lottery gives a polynomial-time selector that ignores both the sampler and the signal. The boundary $c=\gamma_{\mathrm{sig}}^{\mathrm{mix}}(u)$ is also attainable whenever an optimizing lottery can be implemented exactly in the fair-coin randomization model.

For every fixed $c>\gamma_{\mathrm{sig}}^{\mathrm{mix}}(u)$, the reduction underlying \autoref{thm:signalwise-approximation-threshold} can be adapted by independently rerunning the selector and estimating its induced lottery over actions. Amplification makes the probability of any erroneous deletion small over the entire self-reduction. Because the surviving assignments are checked directly, an unsatisfiable instance is never accepted; hence such a selector would imply $\textsc{SAT}\in\mathsf{RP}$ and therefore $\mathsf{RP}=\mathsf{NP}$. Under the standard assumption $\mathsf{RP}\neq\mathsf{NP}$, the signalwise hardness thus persists above $\gamma_{\mathrm{sig}}^{\mathrm{mix}}(u)$. The construction still uses at most two supported signals, so the small-support implication of \autoref{cor:hard-signalwise-small-signal} also holds.

From economic perspective, randomization allows the decision maker to hedge across states without processing the signal, thereby raising the signal-independent benchmark. Once the desired guarantee exceeds the best such hedge, however, the selector must exploit signal-specific information. Requiring it to do so after every supported signal preserves the same qualitative computational barrier.

From a computational perspective, it is widely conjectured that randomization does not enlarge the class of problems solvable in polynomial time, that is, $\mathsf{P}=\mathsf{BPP}$. Under this conjecture, $\mathsf{RP}=\mathsf{P}$, so the assumption $\mathsf{RP}\neq\mathsf{NP}$ used above follows from the standard assumption $\mathsf{P}\neq\mathsf{NP}$. The persistence of the two-signal result has a more direct explanation: amplification repeatedly invokes the selector on the same input and operates only on its internal randomness, leaving the sampler and its signal support unchanged. More broadly, the belief that randomization can often be derandomized provides heuristic support for our conjecture that deterministic ex ante approximation may become tractable under polynomially bounded signal support. 

\subsection{Economic interpretations for computer science framework} 

\paragraph{Worst-case complexity.}
In this paper, we consider the worst-case complexity. The question is whether a \emph{single uniform} polynomial-time selector satisfies the relevant guarantee for every sampler description $q$ and, for the exact and signalwise criteria, every supported signal $s$. Thus, the results do not imply that every sampler is difficult to difficult to process. Particular experiments may admit efficient Bayesian choice because their likelihoods or posterior comparisons have additional structure.

Therefore, the appropriate understanding is that fixing a finite decision environment alone does not guarantee tractable information processing. The results leave substantial room for uniform algorithms on economically or statistically structured promise classes of samplers. The polynomial-support result in \autoref{prop:finite-support-pac-upper} provides one example: a known polynomial bound on $K_q$ restores a randomized polynomial-time PAC selector for fixed approximation and confidence parameters. At the same time, the two-signal hardness results for exact and signalwise selection show that small support alone does not restore tractability under stronger, every-signal guarantees, while the constant-support deterministic ex ante case remains open. More generally, tractability may depend on features such as efficiently evaluable likelihoods, restricted circuit or factorization structure, symmetry, or low-dimensional sufficient statistics. Identifying the relevant structural parameters--and determining how running time and sample complexity depend on them--is a natural direction for further research. Average-case or smoothed analyses could complement the present worst-case results by asking how frequently the hard instances arise in economically relevant classes of experiments.

Another interpretation of our hardness results is that, for any tractable selector, we can always find an experiment where some signal realizations are sufficiently difficult to interpret. Therefore, this suggests a strategic extension where the signal-generating process is itself a design variable. A sender may seek to make decision-relevant information costly for a computationally bounded receiver to extract. One especially sharp question is whether a designer can change this computational burden while holding the induced Blackwell experiment fixed and varying only its representation; the present results do not establish such representation-only obfuscation. Characterizing when a designer can efficiently create and exploit interpretation costs, how they affect information design and equilibrium behavior, and which restrictions on admissible sampler representations eliminate them is a natural direction for future work.

\paragraph{Polynomial and asymptotic running time.}
Our notion of tractability is asymptotic. Let $n=|\langle q\rangle|+|s|$, polynomial time requires the number of elementary operations to be bounded by $Cn^k$ for some constants $C$ and $k$, uniformly over inputs as $n$ grows. It does not imply that a selector is fast on a particular experiment: a large constant or a high polynomial degree may make a polynomial-time algorithm impractical at economically relevant scales. Conversely, an algorithm with an exponential worst-case bound may still be fast on small or favorable instances. Polynomial time also does not by itself provide an economically transparent interpretation or a plausible heuristic structure for the selector.

What this distinction captures is scalability. Polynomial time is a permissive benchmark: it requires only that running time remain bounded by some polynomial as the size of the experiment grows. Our hardness results show that, in the stated hardness regimes, Bayesian signal processing fails even this weak scalability requirement. This makes the conclusion especially strong: the corresponding Bayesian performance requirement is computationally demanding even under a generous notion of tractability, providing a justification for bounded rationality and costly Baysian inference.


\subsection{Sources of bounded rationality and computational hardness}

This subsection provides a detailed summary of how the literature explains the sources of bounded rationality and the origins of computational difficulty in decision making. Then, we clarify how our findings differ from, and contribute to, the existing literature.

The literature surveyed here can be organized around three distinct questions. The first asks why complete substantive rationality is an inadequate description of decision making. \citet{Simon1955} replaces exhaustive maximization with a satisficing search procedure that stops when an aspiration level is met, while \citet{Simon1978} shifts attention from the optimality of an outcome to the feasibility of the process that produces it through the distinction between substantive and procedural rationality. \citet{Conlisk1996} organizes the case for bounded rationality around behavioral evidence, the explanatory success of boundedly rational models, the limits of methodological defenses of unbounded rationality, and the scarcity of deliberation itself. \citet{Radner2000} similarly treats observation, computation, memory, communication, and delay as economically relevant constraints or costs, while distinguishing models that incorporate such costs into otherwise standard optimization from stronger departures associated with genuinely bounded rationality. In this first strand, the central message is that decision procedures and cognitive capacity are themselves economic objects rather than costless background assumptions; these papers explores new possible modeling approaches motivated by the gap between ideal and implementable choice rather than deriving  hardness results for the optimal (Bayesian) decision rule.

The second question asks how a resource-constrained decision maker optimally departs from complete Bayesian inference or unconstrained maximization. \citet{HalpernPass2010,halpern2015algorithmic,AlaouiPenta2022} let the agent choose an algorithm, or decide whether to continue reasoning, by trading the expected improvement in the decision against the cost of obtaining it. \citet{Sims2003,CaplinDean2015,MatejkaMcKay2015} instead make the information structure endogenous: the agent allocates a limited communication capacity or selects a signal structure by balancing its decision value against an entropy-based or more general information cost. \citet{Gabaix2014,Wilson2014} restrict which dimensions enter the agent's internal model or which pieces of past evidence remain available, thereby generating selective attention and systematic updating biases. \citet{LiederGriffiths2020,ZhuGriffiths2026} interpret approximate inference as the solution to an accuracy--resource tradeoff; depending on how cognitive cost is represented, the resulting behavior may take the form of a heuristic inference procedure or a tempered response to new likelihood information. These models generally take a cognitive or informational limitation as a primitive and use it to predict stopping, inattention, sparsity, memory distortions, or approximate updating.

The third question is whether the normative objective itself can be implemented by any efficient algorithm. \citet{Eboli2003} measures the cost of processing information under particular representational architectures, showing how associative memory and Bayesian-network structure generate different computational costs. \citet{Camara2022} studies choice from compactly represented combinatorial menus and shows how separable structure can restore tractability, thereby providing a computational rationale for narrow bracketing. \citet{HazlaEtAl2021} show that fully Bayesian choice in social-learning networks can be hard because an observed action history may encode a combinatorial set of private-signal histories, while particular network structures restore tractability. Classical work in artificial intelligence makes the role of representation especially explicit: \citet{Cooper1990} shows that compact Bayesian networks can encode NP-hard exact-inference problems; \citet{DagumLuby1993} show that relative or additive approximation need not eliminate the encoded distinction; and \citet{Roth1996} connects a broad class of probabilistic and logical reasoning tasks to the hardness of propositional model counting. This strand is closest to our analysis because it asks what a computational procedure can implement, rather than which departure from Bayesian choice is optimal.

Relative to these literatures, our analysis differs along four dimensions. First, we do not posit a particular form of bounded cognition: the prior, utility function, and Bayesian objective remain unchanged, and the question is one of algorithmic implementability rather than the rationalization of behavioral deviations. Second, the economic environment--the state space, action set, prior, and utility matrix--is fixed; the difficulty does not come from a growing action menu, state space, or social network, but from the inference from the signal-generating process and the realized signal. Third, the required output is an action rather than a posterior, so only the location of the posterior relative to an action boundary matters; nevertheless, implicitly represented likelihoods can create a \emph{precision requirement} near that boundary. Fourth, the hardness depends on the form of the guarantee: exact and signalwise hardness can persist with only two supported signals, because the requirement to succeed after \emph{every} supported signal prevents the selector from ignoring a difficult realization. By contrast, deterministic ex ante hardness currently relies on exponentially large support, while a known polynomial support bound restores randomized PAC tractability. Thus, we separate the roles of precision, the every-signal quantifier, and signal-support size.

\section{Conclusions}

In this paper, we develop a computational framework for Bayesian signal processing. Under the standard conjecture that $\mathsf{P}\neq\mathsf{PP}$, we establish a sharp boundary: there is a uniform polynomial-time selector that implements Bayes-optimal choice for every signal-generating process if and only if one action is optimal in every state. This result provides a foundation for boundedly rational choice in real-world decision problems. It also supports models that assign nonnegligible costs to full Bayesian inference.

We also identify how the source of computational hardness changes as the performance guarantee is relaxed. Exact choice can require resolving arbitrarily fine likelihood differences near an action boundary, while the universal quantifier requires the selector to succeed after every supported signal, including the hardest one. Relaxing exactness does not remove the every-signal barrier: any signalwise guarantee above the best signal-independent benchmark remains hard, even with only two supported signals, so the hardness is not necessarily caused by the size of the signal space. Moving to ex ante performance also leaves the problem hard over unrestricted signal-generating processes. For PAC ex ante performance, however, the potentially exponential size of the signal support is the only remaining source of computational difficulty: when the support is polynomially bounded, the Monte Carlo empirical Bayes algorithm provides a tractable solution for any fixed nonexact approximation and confidence level. This is the only general, nontrivial, signal-dependent tractability result established in this paper. They also suggest that, in the uniform worst case, knowing the full signal-generating process may offer no tractability advantage over black-box sampling access. Therefore, these results provide a computational justification for sample-based methods in economics, computer science, and statistics.

Our framework opens the door to a broader analysis of the tractability of optimal and approximately optimal Bayesian decision making and the sources of their computational hardness. One natural direction for future work is to identify restricted classes of samplers that admit efficient selectors implementing Bayes-optimal or approximately optimal choices and to characterize the structural properties that restore such tractability. A particularly immediate open question is whether bounded signal support also restores deterministic ex ante tractability where we conjecture that it does.

\bibliographystyle{plainnat}
\bibliography{signal}

\appendix
\section{Omitted Proofs}\label{appendix:proof}

\begin{proof}[Proof of {\hyperref[thm:exact-bayes-dichotomy]{Theorem~\ref*{thm:exact-bayes-dichotomy}}}]
\mbox{}\par

Suppose $D(u)\neq\emptyset$ and choose $d\in D(u)$.  For all $a\in A$, $q$, and $s$,
\[
W_d(q,s)-W_a(q,s)
=\sum_{\theta\in\Theta}\mu_0(\theta)q(s\mid\theta)
\bigl[u(d,\theta)-u(a,\theta)\bigr]\geq0.
\]
Thus the constant selector $(q,s)\mapsto d$ is exactly optimal.

For the converse, suppose $D(u)=\emptyset$.

\noindent\textbf{Step 1: Construct a rational posterior line crossing a decision boundary.}

By the standard rationalizability result, together with arbitrarily small generic rational perturbations, there are payoff-distinct actions $a$ and $b$ and rational full-support beliefs $\mu_a$ and $\mu_b$ such that the payoff vector of $a$ is uniquely optimal at $\mu_a$ but not optimal at $\mu_b$, while the payoff vector of $b$ is uniquely optimal at $\mu_b$ but not optimal at $\mu_a$; moreover, their connecting segment crosses the finitely many indifference hyperplanes one at a time.

Along the segment $(1-\tau)\mu_a+\tau\mu_b$, the optimal payoff vector must therefore change.  Since the endpoints, utilities, and indifference hyperplanes are rational, every crossing parameter is rational.  Choose such a parameter $t\in(0,1)\cap\mathbb Q$ and let
\[
p^0:=(1-t)\mu_a+t\mu_b,
\qquad
h:=\mu_b-\mu_a.
\]
Let $a^-$ and $a^+$ have the distinct optimal payoff vectors immediately before and after this crossing.  For a sufficiently small rational $\varepsilon>0$, define
\[
p^-=p^0-\varepsilon h,
\qquad
p^+=p^0+\varepsilon h.
\]
Then $p^-$ and $p^+$ are rational and have full support, and the payoff vectors of $a^-$ and $a^+$ are uniquely optimal on the two corresponding open half-segments.  Hence, for $z\in[0,1]$,
\begin{equation}
\label{eq:posterior-line}
\begin{aligned}
p(z)
&=(1-z)p^-+zp^+\\
&=p^0+(2z-1)\varepsilon h.
\end{aligned}
\end{equation}
The unique optimal payoff vector along the two open half-segments is therefore summarized by
\begin{equation}
\label{eq:posterior-winners}
\left\{
u(c,\cdot):
c\in\operatorname*{arg\,max}_{a\in A}\sum_{\theta\in\Theta}p_\theta u(a,\theta)
\right\}
=
\begin{cases}
\{u(a^-,\cdot)\}, & 0\leq z<1/2,\\
\{u(a^+,\cdot)\}, & 1/2<z\leq1.
\end{cases}
\end{equation}
At $z=1/2$, the posterior is exactly $p^0$, where both $a^-$ and $a^+$ are optimal, possibly along with other actions. 

\noindent\textbf{Step 2: Build a fair-coin sampler that realizes the posterior line.}

To make the posterior after a target signal proportional to $p^\pm$, the statewise likelihood of that signal should be proportional to $p_\theta^\pm/\mu_0(\theta)$.  Define the fixed positive rational numbers
\[
\rho_\theta^-:=\frac{p_\theta^-}{\mu_0(\theta)},
\qquad
\rho_\theta^+:=\frac{p_\theta^+}{\mu_0(\theta)}.
\]
Choose a positive integer $M$ divisible by the reduced denominator of every $\rho_\theta^\pm$.  Then
\[
k_\theta^\pm:=M\rho_\theta^\pm
\]
is a positive integer.  Choose a fixed integer $L$ large enough that
\[
k_\theta^\pm\leq2^L
\quad\text{for every }\theta\in\Theta,
\]
and set
\[
\lambda:=\frac{M}{2^L},
\qquad
e_\theta^\pm
:=\lambda\rho_\theta^\pm
=\frac{k_\theta^\pm}{2^L}.
\]
Thus $e_\theta^\pm\in(0,1]$.  If $U$ is the integer represented by $L$ independent fair bits, then
\[
U\sim\operatorname{Unif}\{0,\ldots,2^L-1\},
\qquad
\Pr(U<k_\theta^\pm)=e_\theta^\pm.
\]

Fix two distinct one-bit signals, a target signal $s^\star$ and an alternative signal $s^\circ$.  Let $R$ be the random seed of any Boolean fair-coin subcircuit $B$, and let $b:=B(R)$.  Draw the $L$ bits defining $U$ independently of $R$.  Write
\[
z:=\Pr_R(B(R)=1).
\]
The number $z$ is not supplied to the sampler and need not be computed; it is only the acceptance probability induced by the random seed $R$.  On input $\theta$, define the sampler $q_B$ by
\begin{equation}
\label{eq:sampler-rule}
q_B(\theta;R,U)=
\begin{cases}
s^\star, & b=0\text{ and }U<k_\theta^-,\\
s^\star, & b=1\text{ and }U<k_\theta^+,\\
s^\circ, & \text{otherwise}.
\end{cases}
\end{equation}
Consequently,
\[
q_B(s^\star\mid\theta)
=(1-z)e_\theta^-+ze_\theta^+.
\]
The choice of the scale $\lambda$ makes the ex ante probability of $s^\star$ independent of $z$.  Indeed,
\[
\sum_{\theta\in\Theta}\mu_0(\theta)e_\theta^\pm
=\lambda\sum_{\theta\in\Theta}p_\theta^\pm
=\lambda,
\]
and therefore
\begin{equation}
\label{eq:sampler-mass}
\begin{aligned}
Z(q_B,s^\star)
&=\sum_{\theta\in\Theta}
  \mu_0(\theta)q_B(s^\star\mid\theta)\\
&=(1-z)\lambda+z\lambda
=\lambda>0.
\end{aligned}
\end{equation}
By Bayes' rule, for every $\theta$,
\begin{equation}
\label{eq:sampler-posterior}
\begin{aligned}
\Pr(\theta\mid s^\star)
&=\frac{\mu_0(\theta)q_B(s^\star\mid\theta)}
        {Z(q_B,s^\star)}\\
&=(1-z)p_\theta^-+zp_\theta^+
=p_\theta(z).
\end{aligned}
\end{equation}
Comparing \eqref{eq:sampler-posterior} with \eqref{eq:posterior-line}, the acceptance probability $z$ of $B$ is exactly the coordinate on that posterior line.

\noindent\textbf{Step 3: Reduction.}

We reduce from strict \textsc{Majsat}. Let $\psi(x)$ be an input formula on $n$ variables and let
\[
S:=\#\operatorname{SAT}(\psi).
\]
Introduce a new variable $y$ and construct a padded formula as follows.
\[
\varphi(x,y)
:=\psi(x)\land(y\lor\neg y).
\]
It has $m=n+1$ variables and we have 
\[
N:=\#\operatorname{SAT}(\varphi)=2S.
\]
Hence
\[
S>2^{n-1}
\quad\Longleftrightarrow\quad
N>2^{m-1},
\]
so the padding preserves strict-majority membership.

Draw two fair bits encoding
\[
J\in\{00,01,10,11\}
\]
and independently draw a uniform assignment $X\in\{0,1\}^m$.  Define
\begin{equation}
\label{eq:majority-circuit}
B_\varphi(J,X)=
\begin{cases}
\varphi(X), & J=00,\\
\varphi(X), & J=01,\\
\mathbf{1}\{X\neq0^m\}, & J=10,\\
0, & J=11.
\end{cases}
\end{equation}
Among the $2^{m+2}$ equally likely pairs $(J,X)$, the first two branches contribute $N$ accepting seeds each, the third contributes $2^m-1$, and the fourth contributes none.  Therefore
\[
z_\varphi:=\Pr(B_\varphi=1)
=\frac{2N+2^m-1}{2^{m+2}},
\]
and
\[
z_\varphi-\frac12
=\frac{2N-2^m-1}{2^{m+2}}.
\]
Because $m\geq1$, the numerator is odd and cannot be zero.  Furthermore,
\[
2N-2^m-1>0
\quad\Longleftrightarrow\quad
N\geq2^{m-1}+1
\quad\Longleftrightarrow\quad
N>2^{m-1}.
\]
It follows that
\begin{equation}
\label{eq:majority-side}
\begin{aligned}
z_\varphi&\neq\frac12,\\
\psi\in\textsc{Majsat}
&\quad\Longleftrightarrow\quad z_\varphi>\frac12,\\
\psi\notin\textsc{Majsat}
&\quad\Longleftrightarrow\quad z_\varphi<\frac12.
\end{aligned}
\end{equation}

Use the circuit $B_\varphi$ defined in \eqref{eq:majority-circuit} as the control subcircuit in the sampler rule \eqref{eq:sampler-rule}.  Explicitly, on input state $\theta$, the constructed circuit $q_\varphi$ draws the independent fair random bits $(J,X,U)$ specified in the preceding definitions, computes $B_\varphi(J,X)$, and outputs
\[
q_\varphi(\theta;J,X,U)=
\begin{cases}
s^\star,
  & B_\varphi(J,X)=0\text{ and }U<k_\theta^-,\\
s^\star,
  & B_\varphi(J,X)=1\text{ and }U<k_\theta^+,\\
s^\circ, & \text{otherwise}.
\end{cases}
\]
This display is the complete definition of the sampler produced from $\psi$.  Averaging over its random bits gives
\[
q_\varphi(s^\star\mid\theta)
=(1-z_\varphi)e_\theta^-
+z_\varphi e_\theta^+,
\qquad
q_\varphi(s^\circ\mid\theta)
=1-q_\varphi(s^\star\mid\theta).
\]
Applying \eqref{eq:sampler-mass} and \eqref{eq:sampler-posterior} with $B=B_\varphi$ gives
\[
Z(q_\varphi,s^\star)=\lambda>0
\]
and the posterior conditional on the target signal is
\[
\Pr(\theta\mid s^\star)
=(1-z_\varphi)p_\theta^-
+z_\varphi p_\theta^+
=p_\theta(z_\varphi).
\]
Equations~\eqref{eq:posterior-winners} and~\eqref{eq:majority-side} yield
\[
\begin{array}{lll}
\psi\in\textsc{Majsat}
&\Longrightarrow z_\varphi>1/2
&\Longrightarrow u(a^+,\cdot)\text{ is the unique optimal payoff vector},\\[1mm]
\psi\notin\textsc{Majsat}
&\Longrightarrow z_\varphi<1/2
&\Longrightarrow u(a^-,\cdot)\text{ is the unique optimal payoff vector}.
\end{array}
\]

It remains to verify that this is a polynomial-time construction.  All objects obtained from the fixed environment--including $a^-$, $a^+$, $p^-$, $p^+$, $M$, $L$, $\lambda$, the thresholds $k_\theta^\pm$, and the two signals--are hard-wired constants.  Given $\psi$, the padding takes linear time.  The circuit contains one copy of $\varphi$, constant-size branch-selection logic, $O(m)$ gates for the test $X\neq0^m$, and the fixed state-dependent threshold circuit.  It uses exactly
\[
m+2+L
\]
fair random bits.  Its description size is
\[
O\bigl(|\psi|+m+|\Theta|L\bigr)=O(|\psi|),
\]
because $m\leq|\psi|+1$ and $\Theta$ and $L$ are fixed.  The gate list and wiring can therefore be written in polynomial time, and $s^\star$ has constant length.  Crucially, the constructor never computes either $S$ or $N$; those counts appear only in the analysis of the circuit's acceptance probability.

Define
\[
F(\psi):=(\langle q_\varphi\rangle,s^\star).
\]
The preceding paragraph proves that $F$ is computable in polynomial time.  Now let $\alpha$ be any exact Bayes selector and define the fixed output decoder
\[
g(a):=\mathbf{1}\{u(a,\cdot)=u(a^+,\cdot)\}.
\]
Since the environment and $a^+$ are fixed, $g$ runs in constant time. By construction, we have that
\[
g\!\left(
  \alpha(\langle q_\varphi\rangle,s^\star)
\right)=1
\quad\Longleftrightarrow\quad
\psi\in\textsc{Majsat}.
\]
Thus one call to any exact selector, followed by the fixed decoder $g$, solves the $\mathsf{PP}$-hard source problem.  Computing an exact Bayes action is therefore $\mathsf{PP}$-hard under a polynomial-time one-call output-decoding reduction. Finally, if $D(u)=\emptyset$ and a polynomial-time exact selector existed, $\mathsf{PP}$-hardness would imply $\mathsf{PP}\subseteq\mathsf{P}$ and hence $\mathsf{P}=\mathsf{PP}$, a contradiction.
\end{proof}

\begin{proof}[Proof of \autoref{thm:signalwise-approximation-threshold}]
First suppose that $c\leq\gamma_{\mathrm{sig}}(u)$, and choose an action $d$ attaining the maximum in \eqref{eq:gamma-sig}.  This maximum is attained because $A$ is finite.  

For every $\theta\in\Theta$, the definition of $d$ gives
$u(d,\theta)\geq\gamma_{\mathrm{sig}}(u)M_\theta$. Consequently,
\begin{equation*}
\begin{aligned}
W_d(q,s)
&=\sum_{\theta\in\Theta}\mu_0(\theta)q(s\mid\theta)u(d,\theta)\\
&\geq
\gamma_{\mathrm{sig}}(u)
\sum_{\theta\in\Theta}\mu_0(\theta)q(s\mid\theta) M_\theta\\
&\geq
\gamma_{\mathrm{sig}}(u)
\max_{a\in A}W_a(q,s).
\end{aligned}
\end{equation*}
Thus the constant selector that always returns $d$ achieves every factor $c\leq\gamma_{\mathrm{sig}}(u)$.

For the converse, fix $c>\gamma_{\mathrm{sig}}(u)$ and suppose that a polynomial-time deterministic signalwise $c$-selector $\alpha$ exists.  We show, step by step, how to use $\alpha$ as an oracle to decide an arbitrary instance of \textsc{SAT}, which is a well-known $\mathsf{NP}$-complete problem.  

Write
\[
\Theta=\{\theta_1,\ldots,\theta_m\}.
\]
The hard case necessarily has $m\geq2$.  Indeed, if $m=1$, an action attaining $M_{\theta_1}$ makes \eqref{eq:gamma-sig} equal to one, which is incompatible with $c\leq1$ and $c>\gamma_{\mathrm{sig}}(u)$.

\noindent\textbf{Step 1: Construct an $m$-to-$(m-1)$ SAT list compressor.}

Suppose we are given $m$ residual formulas
\[
\boldsymbol\psi=(\psi_1,\ldots,\psi_m)
\]
over the same declared variables $y_1,\ldots,y_r$. The compressor will delete one formula while preserving whether at least one formula in the list is satisfiable.

Fix two distinct one-bit signals $s^\star$ and $s^\circ$.  Construct a fair-coin sampler $q_{\boldsymbol\psi}$ that draws
\[
R\sim\operatorname{Unif}\{0,1\}^r
\]
and, on input state $\theta$, applies the rule
\[
q_{\boldsymbol\psi}(\theta;R)
=
\begin{cases}
s^\star,
  & \theta=\theta_i\text{ for some }i\in\{1,\ldots,m\}
    \text{ and }\psi_i(R)=1,\\
s^\circ, & \text{otherwise}.
\end{cases}
\]
In state $\theta_i$, the target signal is produced by exactly the satisfying assignments of $\psi_i$, and therefore
\begin{equation}
\label{eq:sat-compressor-likelihood}
q_{\boldsymbol\psi}(s^\star\mid\theta_i)
=
\frac{\#\operatorname{SAT}(\psi_i)}{2^r}.
\end{equation}
Since $m$ is fixed by the environment, the sampler contains only a constant number of formula circuits and its description is polynomial in the total description length of the residual formulas.  The same construction also covers $r=0$, with $\{0,1\}^0$ interpreted as a singleton.

Run the selector on the target signal and choose
\[
a=\alpha(\langle q_{\boldsymbol\psi}\rangle,s^\star),
\qquad
j\in\arg\min_{1\leq i\leq m}
\frac{u(a,\theta_i)}{M_{\theta_i}}.
\]
The index $j$ is computable by a constant-size lookup in the fixed utility table.  By the definition of \eqref{eq:gamma-sig},
\begin{equation}
\label{eq:sat-compressor-bad-state}
\frac{u(a,\theta_j)}{M_{\theta_j}}
=\min_{1\leq i\leq m}
  \frac{u(a,\theta_i)}{M_{\theta_i}}
\leq
\max_{b\in A}\min_{1\leq i\leq m}
  \frac{u(b,\theta_i)}{M_{\theta_i}}
=\gamma_{\mathrm{sig}}(u)
<c.
\end{equation}
The compressor deletes $\psi_j$.

\noindent\textbf{Step 2: Prove that the deletion is safe.}

We verify that
\begin{equation}
\label{eq:sat-compressor-preserves-or}
\bigvee_{i=1}^m\psi_i\text{ is satisfiable}
\quad\Longleftrightarrow\quad
\bigvee_{i\neq j}\psi_i\text{ is satisfiable}.
\end{equation}
There are three cases.

\textbf{Case 1:} If no formula is satisfiable, both sides of \eqref{eq:sat-compressor-preserves-or} are false, regardless of which formula is deleted. 

\textbf{Case 2:} If at least two formulas are satisfiable, deleting one formula leaves at least one satisfiable formula, so both sides of \eqref{eq:sat-compressor-preserves-or} are true.

\textbf{Case 3:} It remains to consider the case in which exactly one formula, say $\psi_k$, is satisfiable.  We claim that $k\neq j$.  Suppose instead that $k=j$.  

By \eqref{eq:sat-compressor-likelihood},
$q_{\boldsymbol\psi}(s^\star\mid\theta_i)=0$ for every $i\neq j$, whereas $q_{\boldsymbol\psi}(s^\star\mid\theta_j)>0$. Since $\mu_0$ has full support, $s^\star$ perfectly reveals $\theta_j$, so \eqref{eq:sat-compressor-bad-state} yields
\[
\frac{W_a(q_{\boldsymbol\psi},s^\star)}
     {\max_{b\in A}W_b(q_{\boldsymbol\psi},s^\star)}
=
\frac{u(a,\theta_j)}{M_{\theta_j}}
<c.
\]
This contradicts the signalwise guarantee \eqref{eq:signalwise-guarantee}.  Thus $k\neq j$: the unique satisfiable formula $\psi_k$ remains after $\psi_j$ is deleted.  This establishes the third case and proves \eqref{eq:sat-compressor-preserves-or}.

\noindent\textbf{Step 3: Use the compressor to decide an arbitrary SAT instance.}

Let
\[
\Phi(y_1,\ldots,y_n)
\]
be an arbitrary Boolean formula.  After assigning the first $k$ variables, a prefix is a string
\[
\sigma=(\sigma_1,\ldots,\sigma_k)\in\{0,1\}^k,
\]
and its residual formula is
\[
\Phi_\sigma(y_{k+1},\ldots,y_n)
:=\Phi(\sigma_1,\ldots,\sigma_k,y_{k+1},\ldots,y_n).
\]
At depth $k$, maintain a list $\mathcal L_k$ of prefixes satisfying the two invariants
\begin{equation}
\label{eq:sat-list-invariant}
|\mathcal L_k|\leq m-1,
\qquad
\Phi\text{ is satisfiable}
\quad\Longleftrightarrow\quad
\text{some }\Phi_\sigma\text{ with }\sigma\in\mathcal L_k
\text{ is satisfiable}.
\end{equation}
Initialize
\[
\mathcal L_0=\{\varepsilon\},
\qquad
\Phi_\varepsilon=\Phi,
\]
The invariant holds initially, and its size bound is valid because $m\geq2$.

Suppose $\mathcal L_k$ has been constructed for some $k<n$, we will construct $\mathcal L_{k+1}$. Let
\[
\widetilde{\mathcal L}_{k+1}
:=
\bigcup_{\sigma\in\mathcal L_k}\{\sigma0,\sigma1\}.
\]
For each prefix $\sigma$,
\[
\Phi_\sigma\text{ is satisfiable}
\quad\Longleftrightarrow\quad
\Phi_{\sigma0}\text{ or }\Phi_{\sigma1}\text{ is satisfiable}.
\]
Thus the expansion preserves the second invariant in \eqref{eq:sat-list-invariant}, and
\[
|\widetilde{\mathcal L}_{k+1}| = 2|\mathcal L_k|
\leq2(m-1).
\]

Now, consider the following procedure to reduce the size of $\widetilde{\mathcal L}_{k+1}$. While $|\widetilde{\mathcal L}_{k+1}|\geq m$, choose any $m$ distinct prefixes
\[
\sigma^{(1)},\ldots,\sigma^{(m)}
\in\widetilde{\mathcal L}_{k+1}
\]
and feed their residual formulas
\[
\psi_i:=\Phi_{\sigma^{(i)}},
\qquad i=1,\ldots,m,
\]
to the compressor constructed in Step 1.  If the compressor returns the deletion index $j$, remove the prefix $\sigma^{(j)}$ from the full list. By Step 2, each deletion preserves the second invariant for the entire list.

Each oracle call removes one prefix.  Starting from at most $2(m-1)$ prefixes, at most
\[
2(m-1)-(m-1)=m-1
\]
calls reduce the list to size at most $m-1$.  Set the resulting list $\mathcal L_{k+1}$.  This completes the inductive construction.

At depth $n$, every prefix in $\mathcal L_n$ is a complete assignment and its residual formula has no remaining variables.  Directly evaluate $\Phi(\sigma)$ for each of the at most $m-1$ surviving assignments.  The invariant \eqref{eq:sat-list-invariant} implies that one of them satisfies $\Phi$ if and only if the original formula is satisfiable.

Finally, this procedure is polynomial time.  There are at most $m-1$ selector calls per variable and hence at most $n(m-1)=O(n)$ calls in total, because $m$ is fixed by the environment.  A residual formula is represented by hard-wiring a prefix into the original formula, so its size remains polynomial in $|\Phi|$.  Each queried sampler contains at most the fixed number $m$ of such formula circuits, at most $n$ fair random bits, and constant-size state and output logic.  Its description can therefore be constructed in polynomial time.  All later queries may depend on earlier outputs of $\alpha$, which is why this is a polynomial-time reduction.

We have shown that a polynomial-time signalwise $c$-selector for $c>\gamma_{\mathrm{sig}}(u)$ would put \textsc{SAT} in $\mathsf{P}$ and imply $\mathsf{P}=\mathsf{NP}$.  Under $\mathsf{P}\neq\mathsf{NP}$, no such selector exists.
\end{proof}

\begin{proof}[Proof of \autoref{thm:ex-ante-approximation-threshold}]
Choose a prior-optimal action
\[
a^0\in\arg\max_{a\in A}
\sum_{\theta\in\Theta}\mu_0(\theta)u(a,\theta)
\]
and let $\alpha^0$ be the constant selector defined by
\[
\alpha^0(\langle q\rangle,s):=a^0.
\]
We know that
\begin{equation}
\label{eq:constant-selector-value}
\operatorname{VAL}_{\alpha^0}(q)=V^{\mathrm{NI}}
\end{equation}
for every $q$.  On the other hand,
\begin{equation}
\label{eq:opt-bounded-by-full-information}
\begin{aligned}
\operatorname{OPT}(q)
&\leq
\sum_s\sum_{\theta\in\Theta}
\mu_0(\theta)q(s\mid\theta)M_\theta\\
&=V^{\mathrm{FI}}.
\end{aligned}
\end{equation}
If $c\leq\gamma_{\mathrm{EA}}(u,\mu_0)$, by equations~\eqref{eq:constant-selector-value} and \eqref{eq:opt-bounded-by-full-information}, 
\[
\operatorname{VAL}_{\alpha^0}(q)
=V^{\mathrm{NI}}
=\gamma_{\mathrm{EA}}(u,\mu_0)V^{\mathrm{FI}}
\geq\gamma_{\mathrm{EA}}(u,\mu_0)\operatorname{OPT}(q)
\geq c\,\operatorname{OPT}(q).
\]
This proves the ``if'' direction.

For the converse, suppose that $c>\gamma_{\mathrm{EA}}(u,\mu_0)$ and that a uniform deterministic polynomial-time ex ante $c$-selector $\alpha$ exists.  Choose a rational number $\bar c$ with
\[
\gamma_{\mathrm{EA}}(u,\mu_0)<\bar c<c.
\]
The selector $\alpha$ is also an ex ante $\bar c$-selector. Define
\[
\delta
:=
\bar cV^{\mathrm{FI}}-V^{\mathrm{NI}}
=
\bigl(\bar c-\gamma_{\mathrm{EA}}(u,\mu_0)\bigr)V^{\mathrm{FI}}
>0.
\]
We use $\alpha$ to give a $\mathsf{BPP}$ algorithm for Statistical Difference, which is a $\mathsf{SZK}$-complete problem \citep{SahaiVadhan2003}.  The reduction constructs a sampler $q$ with the following two properties.  When the two input distributions are far, the signal almost reveals the state, so the approximation guarantee forces $\operatorname{VAL}_\alpha(q)$ to be large.  When they are close, the action chosen by $\alpha$ is almost independent of the state, so $\operatorname{VAL}_\alpha(q)$ is at most $V^{\mathrm{NI}}$ up to a small error. 

\noindent\textbf{Step 1: From Polarized Statistical Difference to construct the sampler.}

For two distributions generated by sampler circuits $P_0$ and $P_1$, write
\[
d_{\mathrm{TV}}(P_0,P_1)
:=
\frac{1}{2}\sum_y
\left|\Pr[P_0=y]-\Pr[P_1=y]\right|.
\]
The standard Statistical Difference problem asks us to distinguish whether 
\[
d_{\mathrm{TV}}(P_0,P_1) >\frac{2}{3} \qquad \text{or} \qquad d_{\mathrm{TV}}(P_0,P_1) <\frac{1}{3} 
\]
in polynomial time. 

Let $t:=|\Theta|$.  If $t=1$, then $V^{\mathrm{NI}}=V^{\mathrm{FI}}$ and $\gamma_{\mathrm{EA}}(u,\mu_0)=1$, which is incompatible with $\bar c<c\leq1$ and $\bar c>\gamma_{\mathrm{EA}}(u,\mu_0)$.  Hence $t\geq2$. Let $\ell:=\lceil\log_2t\rceil$. Assign each state $\theta$ a distinct binary codeword
\[
\kappa(\theta)
=
\bigl(\kappa_1(\theta),\ldots,\kappa_\ell(\theta)\bigr)
\in\{0,1\}^\ell.
\]
Choose a fixed rational $\varepsilon>0$ so small that
\[
\ell\varepsilon(\bar c+1)V^{\mathrm{FI}}
<\frac{\delta}{2}.
\]
Apply Polarization Lemma \citep[Lemma 3.3]{SahaiVadhan2003} and relabel the resulting sampler circuits as $P_0$ and $P_1$.  For this choice of $\varepsilon$, it suffices to distinguish the promise
\[
d_{\mathrm{TV}}(P_0,P_1)\leq\varepsilon
\qquad\text{or}\qquad
d_{\mathrm{TV}}(P_0,P_1)\geq1-\varepsilon.
\]
We call these the close and far cases, respectively.

From $P_0$ and $P_1$, construct a sampler $q$ that, in state $\theta$, independently draws $\ell$ samples and outputs
\[
S=(Y_1,\ldots,Y_\ell),
\qquad
Y_k\sim
\begin{cases}
P_0, & \kappa_k(\theta)=0,\\
P_1, & \kappa_k(\theta)=1
\end{cases}
\qquad(k=1,\ldots,\ell).
\]
Since the environment and $\ell$ are fixed, the circuit for $q$ has size polynomial in the descriptions of $P_0$ and $P_1$.

\noindent\textbf{Step 2: In the far case, the optimum is close to full information.}

Suppose that $d_{\mathrm{TV}}(P_0,P_1)\geq1-\varepsilon$.  Let
\[
E:=\{y:\Pr[P_0=y]\geq\Pr[P_1=y]\}.
\]
By the property of total variation distance, we have that 
\[
P_0(E)-P_1(E)
=
d_{\mathrm{TV}}(P_0,P_1).
\]
Therefore, 
\[
P_0(E^c)+P_1(E)
=
1-d_{\mathrm{TV}}(P_0,P_1)
\leq\varepsilon.
\]

Given a signal $S = (Y_1,\ldots,Y_\ell)$, construct
\[
\widehat b_k
:=
\begin{cases}
0, & Y_k\in E,\\
1, & Y_k\notin E.
\end{cases}
\]
Conditional on state $\theta$,
\[
\Pr\!\left[\widehat b_k\neq\kappa_k(\theta)\mid\theta\right]
=
\begin{cases}
P_0(E^c) & \kappa_k(\theta)=0,\\
P_1(E) & \kappa_k(\theta)=1
\end{cases}
\leq\varepsilon.
\]
By union bound, we have 
\[
\Pr\!\left[
(\widehat b_1,\ldots,\widehat b_\ell)=\kappa(\theta)
\mid\theta
\right]
\geq1-\ell\varepsilon.
\]

We now construct a possibly inefficient response rule $\beta$.  Fix
\[
a^\star_\theta\in\arg\max_{a\in A}u(a,\theta)
\qquad\text{for every }\theta,
\]
and fix an arbitrary fallback action $a^\circ\in A$.  Given a signal $S=(Y_1,\ldots,Y_\ell)$, first compute the decoded string $\widehat b=(\widehat b_1,\ldots,\widehat b_\ell)$ and then set
\[
\beta(S)
:=
\begin{cases}
a^\star_{\widehat\theta},
  & \widehat b=\kappa(\widehat\theta)
    \text{ for some }\widehat\theta\in\Theta,\\
a^\circ,
  & \widehat b\notin\kappa(\Theta).
\end{cases}
\]
The codewords are distinct, so the state $\widehat\theta$ in the first line is unique.  If the true state is $\theta$, then on the event $\widehat b=\kappa(\theta)$ the rule returns $a^\star_\theta$ and obtains $M_\theta$.  Dropping the positive payoff on the complementary event gives
\[
\mathbb E_{S\sim q(\cdot\mid\theta)}
\bigl[u(\beta(S),\theta)\bigr]
\geq
(1-\ell\varepsilon)M_\theta.
\]
Because the unrestricted Bayes optimum maximizes signal by signal, it weakly dominates $\beta$:
\begin{equation}
\label{eq:far-case-optimum}
\begin{aligned}
\operatorname{OPT}(q)
&=\sum_s\max_{a\in A}W_a(q,s)\\
&\geq\sum_s W_{\beta(s)}(q,s)\\
&=\sum_{\theta\in\Theta}\mu_0(\theta)
  \mathbb E_{S\sim q(\cdot\mid\theta)}
  \bigl[u(\beta(S),\theta)\bigr]\\
&\geq(1-\ell\varepsilon)
  \sum_{\theta\in\Theta}\mu_0(\theta)M_\theta\\
&=(1-\ell\varepsilon)V^{\mathrm{FI}}.
\end{aligned}
\end{equation}
The $\bar c$-guarantee for $\alpha$ now implies
\[
\operatorname{VAL}_\alpha(q)
\geq
\bar c\,\operatorname{OPT}(q)
\geq
\bar c(1-\ell\varepsilon)V^{\mathrm{FI}}.
\]

\noindent\textbf{Step 3: In the close case, the selector acts nearly no-information.}

Suppose instead that $d_{\mathrm{TV}}(P_0,P_1)\leq\varepsilon$.  Let
\[
Q_\theta
:=
\bigotimes_{k=1}^\ell P_{\kappa_k(\theta)}
\qquad\text{and}\qquad
\overline{Q}:=P_0^{\otimes\ell}.
\]
These are, respectively, the signal distribution in state $\theta$ and a common reference distribution. By the product inequality for total variation distance \citep[Lemma 3.6.5]{Durrett2019},
\[
d_{\mathrm{TV}}(Q_\theta,\overline{Q})
\leq
\sum_{k:\,\kappa_k(\theta)=1}
d_{\mathrm{TV}}(P_1,P_0)
\leq\ell\varepsilon.
\]

For the fixed sampler $q$, write
\[
f_q(S):=\alpha(\langle q\rangle,S).
\]
Let $\pi_\theta$ be the distribution of $f_q(S)$ under $S\sim Q_\theta$, and let $\overline{\pi}$ be its distribution under $S\sim\overline{Q}$. By data-processing inequality,
\[
d_{\mathrm{TV}}(\pi_\theta,\overline{\pi})
\leq
d_{\mathrm{TV}}(Q_\theta,\overline Q)
\leq \ell\varepsilon.
\]
Since $0<u(a,\theta)\leq M_\theta$, it follows state by state that
\[
\mathbb E_{a\sim\pi_\theta}[u(a,\theta)]
\leq
\mathbb E_{a\sim\overline{\pi}}[u(a,\theta)]
+\ell\varepsilon M_\theta.
\]
Consequently,
\begin{equation}
\label{eq:close-case-selector-value}
\begin{aligned}
\operatorname{VAL}_\alpha(q)
&=\sum_{\theta\in\Theta}\mu_0(\theta)
  \mathbb E_{a\sim\pi_\theta}[u(a,\theta)]\\
&\leq
\sum_{a\in A}\overline{\pi}(a)
  \sum_{\theta\in\Theta}\mu_0(\theta)u(a,\theta)
  +\ell\varepsilon V^{\mathrm{FI}}\\
&\leq
V^{\mathrm{NI}}+\ell\varepsilon V^{\mathrm{FI}}.
\end{aligned}
\end{equation}
The last inequality holds because the first term is a convex combination, with weights $\overline{\pi}(a)$, of the prior payoffs of fixed actions, each of which is at most $V^{\mathrm{NI}}$.

\noindent\textbf{Step 4: Obtain a fixed gap between the two cases.}

Define the far-case lower bound and the close-case upper bound by
\[
\underline{V}_{\mathrm{far}}
:=
\bar c(1-\ell\varepsilon)V^{\mathrm{FI}},
\qquad
\overline{V}_{\mathrm{close}}
:=
V^{\mathrm{NI}}+\ell\varepsilon V^{\mathrm{FI}}.
\]
Their difference is
\[
\underline{V}_{\mathrm{far}}-\overline{V}_{\mathrm{close}}
=\delta-\ell\varepsilon(\bar c+1)V^{\mathrm{FI}}
>\frac{\delta}{2}>0.
\]
Thus $\operatorname{VAL}_\alpha(q)\geq\underline{V}_{\mathrm{far}}$ in the far case, whereas $\operatorname{VAL}_\alpha(q)\leq\overline{V}_{\mathrm{close}}$ in the close case.  The midpoint
\[
\tau
:=
\frac{\underline{V}_{\mathrm{far}}+\overline{V}_{\mathrm{close}}}{2}
\]
separates the cases with margin greater than $\delta/4$.

\noindent\textbf{Step 5: Estimate the value and decide Statistical Difference.}

It remains only to estimate $\operatorname{VAL}_\alpha(q)$ without enumerating the signal space.  For a positive integer $N_{\mathrm{samp}}$ to be chosen below and for each state $\theta$, draw $S_{\theta,1},\ldots,S_{\theta,N_{\mathrm{samp}}}$ independently from $q(\cdot\mid\theta)$ and define the true and empirical statewise payoffs
\[
v_\theta
:=
\mathbb E_{S\sim q(\cdot\mid\theta)}
\!\left[u\!\left(\alpha(\langle q\rangle,S),\theta\right)\right],
\qquad
\widehat v_\theta
:=
\frac{1}{N_{\mathrm{samp}}}
\sum_{r=1}^{N_{\mathrm{samp}}}
u\!\left(\alpha(\langle q\rangle,S_{\theta,r}),\theta\right).
\]
Let
\[
\widehat V
:=
\sum_{\theta\in\Theta}\mu_0(\theta)\widehat v_\theta.
\]
Since
\[
\operatorname{VAL}_\alpha(q)
=
\sum_{\theta\in\Theta}\mu_0(\theta)v_\theta,
\]
the estimator $\widehat V$ is unbiased.  Let
\[
u_{\max}
:=
\max_{a\in A,\,\theta\in\Theta}u(a,\theta).
\]
For every fixed $\theta$, the $N_{\mathrm{samp}}$ realized payoffs in $\widehat v_\theta$ are independent random variables in $[0,u_{\max}]$.  By Hoeffding's inequality,
\[
\Pr\!\left[
\left|\widehat v_\theta-v_\theta\right|
\geq\frac{\delta}{8}
\right]
\leq
2\exp\!\left(
-\frac{N_{\mathrm{samp}}\delta^2}{32u_{\max}^2}
\right).
\]
Because $t=|\Theta|$, a union bound over the states yields
\[
\Pr\!\left[
\max_{\theta\in\Theta}
\left|\widehat v_\theta-v_\theta\right|
\geq\frac{\delta}{8}
\right]
\leq
2t\exp\!\left(
-\frac{N_{\mathrm{samp}}\delta^2}{32u_{\max}^2}
\right).
\]
Choose
\[
N_{\mathrm{samp}}
:=
\left\lceil
\frac{32u_{\max}^2}{\delta^2}\ln(6t)
\right\rceil.
\]
The preceding failure probability is then at most $1/3$.  On the complementary event,
\[
\left|\widehat V-\operatorname{VAL}_\alpha(q)\right|
\leq
\sum_{\theta\in\Theta}\mu_0(\theta)
\left|\widehat v_\theta-v_\theta\right| \leq 
\max_{\theta\in\Theta}
\left|\widehat v_\theta-v_\theta\right| <\frac{\delta}{8}.
\]
Thus $\widehat V$ has the required accuracy with probability at least $2/3$.
Since $\delta/8<\delta/4$, the midpoint comparison from Step~4 then correctly separates the far and close cases.

Finally, $t$, $u_{\max}$, and $\delta$ depend only on the fixed environment and the fixed factor $\bar c$, not on the Statistical Difference input.  Hence $N_{\mathrm{samp}}$ is a fixed integer, and the estimator makes $tN_{\mathrm{samp}}$ sampler calls and the same number of selector calls.  Each draw from $q$ runs the fixed number $\ell$ of polarized sampler circuits, and each call to the uniform selector $\alpha$ takes polynomial time in the description of $q$ and the sampled signal.  The fixed rational payoffs and prior weights add only polynomial-time arithmetic.  Therefore computing $\widehat V$ and declaring the instance far when $\widehat V>\tau$ and close otherwise is a uniform randomized polynomial-time algorithm for polarized Statistical Difference with bounded error at most $\frac{1}{3}$.

Therefore, this algorithm places Statistical Difference in $\mathsf{BPP}$.  By $\mathsf{SZK}$-completeness, $\mathsf{SZK}\subseteq\mathsf{BPP}$.  Since it is well known that $\mathsf{BPP}\subseteq\mathsf{SZK}$, we get $\mathsf{SZK}=\mathsf{BPP}$, contradicting the assumption.
\end{proof}

\begin{proof}[Proof of \autoref{thm:pac-ex-ante-threshold}]
If $c\leq\gamma_{\mathrm{EA}}(u,\mu_0)$, the prior-optimal constant selector $\alpha^0$ from the proof of \autoref{thm:ex-ante-approximation-threshold} remains valid and completes the proof of this part. 

Conversely, suppose that $c>\gamma_{\mathrm{EA}}(u,\mu_0)$ and that a uniform polynomial-time randomized $(c,\delta)$-PAC ex ante selector $\alpha$ exists.  Choose a rational $\bar c\in(\gamma_{\mathrm{EA}}(u,\mu_0),c)$ and set
\[
\Delta:=\bar cV^{\mathrm{FI}}-V^{\mathrm{NI}}>0.
\]
Let $t:=|\Theta|$ and $\ell:=\lceil\log_2t\rceil$; as in the preceding proof, the present case implies $t\geq2$.  Choose a fixed rational $\varepsilon>0$ such that
\[
\ell\varepsilon(\bar c+1)V^{\mathrm{FI}}<\frac{\Delta}{2}.
\]
Apply exactly the polarization, state encoding, and sampler construction from Step~1 of the proof of \autoref{thm:ex-ante-approximation-threshold}.  Thus the polarized circuits satisfy
\[
d_{\mathrm{TV}}(P_0,P_1)\leq\varepsilon
\qquad\text{or}\qquad
d_{\mathrm{TV}}(P_0,P_1)\geq1-\varepsilon,
\]
and the resulting sampler $q$ has polynomial description size.

In the far case, the decoding argument remains valid and we can get that by \eqref{eq:far-case-optimum},
\[
\operatorname{OPT}(q)\geq(1-\ell\varepsilon)V^{\mathrm{FI}}.
\]
By $\bar c$-approximation, the PAC guarantee implies
\[
\Pr_R\!\left[
\operatorname{VAL}_{\alpha_R}(q)\geq
\bar c(1-\ell\varepsilon)V^{\mathrm{FI}}
\right]
\geq1-\delta.
\]
Set $F:=\bar c(1-\ell\varepsilon)V^{\mathrm{FI}}$.

In the close case, fix an arbitrary tape $R$, including one on which the PAC guarantee fails.  The rule $\alpha_R$ is a deterministic total selector, so \eqref{eq:close-case-selector-value} still applies. Hence, for every $R$,
\[
\operatorname{VAL}_{\alpha_R}(q)
\leq V^{\mathrm{NI}}+\ell\varepsilon V^{\mathrm{FI}}
=:C.
\]
Consequently,
\[
G:=F-C
=\Delta-\ell\varepsilon(\bar c+1)V^{\mathrm{FI}}
>\frac{\Delta}{2}>0.
\]
The new issue is that the far-case lower bound holds only for a $1-\delta$ fraction of random tapes, whereas the close-case upper bound holds pointwise for every tape.

To handle this difference, let
\[
m:=\left\lceil\frac{\ln 9}{\ln(1/\delta)}\right\rceil
\]
and draw independent tapes $R_1,\ldots,R_m$.  In the far case, with probability at least $1-\delta^m\geq8/9$, at least one of the fixed-tape rules $\alpha_{R_i}$ has value at least $F$; in the close case, every one of them has value at most $C$.

Similar to Step 5 of the proof of \autoref{thm:ex-ante-approximation-threshold}, let
\[
u_{\max}:=\max_{a\in A,\,\theta\in\Theta}u(a,\theta),
\qquad
N_{\mathrm{val}}
:=\left\lceil
\frac{8u_{\max}^2}{G^2}\ln(18mt)
\right\rceil.
\]
For each $i$ and $\theta$, draw independent samples $S_{i,\theta,1},\ldots,S_{i,\theta,N_{\mathrm{val}}}\sim q(\cdot\mid\theta)$, independently of the tapes, and define
\[
\widehat V_i
:=
\sum_{\theta\in\Theta}\mu_0(\theta)
\frac{1}{N_{\mathrm{val}}}
\sum_{r=1}^{N_{\mathrm{val}}}
u\!\left(
\alpha_{R_i}(\langle q\rangle,S_{i,\theta,r}),\theta
\right).
\]
Conditional on the tapes, Hoeffding's inequality and a union bound over the $mt$ tape--state pairs give
\[
\Pr\!\left[
\max_{1\leq i\leq m}
\left|\widehat V_i-\operatorname{VAL}_{\alpha_{R_i}}(q)\right|
<\frac{G}{4}
\,\middle|\,
R_1,\ldots,R_m
\right]
\geq\frac89.
\]

Let $\tau:=(F+C)/2$ and declare the polarized instance far if and only if $\max_i\widehat V_i>\tau$.  In the far case, a successful tape together with accurate validation gives
\[
\max_i\widehat V_i
\geq F-\frac G4
=\tau+\frac G4
>\tau;
\]
this joint event has probability at least $1-1/9-1/9=7/9$.  In the close case, the pointwise bound for every tape and accurate validation give
\[
\max_i\widehat V_i
\leq C+\frac G4
=\tau-\frac G4
<\tau
\]
with probability at least $8/9$.

The analysis of the size and time complexity of this reduction is the same as the proof of \autoref{thm:ex-ante-approximation-threshold}. The above procedure is therefore a randomized polynomial-time algorithm for polarized Statistical Difference with bounded error, which contracts to $\mathsf{SZK}\neq \mathsf{BPP}$.
\end{proof}

\begin{proof}[Proof of \autoref{prop:finite-support-pac-upper}]
For a policy $h:\mathcal S_q\to A$, define its empirical ex ante value by
\[
\widehat{\operatorname{VAL}}_N(h)
:=
\sum_{\theta\in\Theta}\mu_0(\theta)
\frac{1}{N}\sum_{i=1}^N
u\bigl(h(S_{\theta,i}),\theta\bigr).
\]

Before the main proof, we firstly show the following lemma.
\begin{lemma}
\label{lem:empirical-bayes-erm}
For every realization of the training sample, the restriction of $\widehat\alpha_N$ to $\mathcal S_q$ satisfies
\[
\widehat\alpha_N
\in
\operatorname*{arg\,max}_{h\in\mathcal H_q}
\widehat{\operatorname{VAL}}_N(h),
\qquad
\mathcal H_q:=A^{\mathcal S_q}.
\]
\end{lemma}

\begin{proof}[Proof of \autoref{lem:empirical-bayes-erm}]
Regrouping the empirical objective by signal gives
\[
\widehat{\operatorname{VAL}}_N(h)
=
\sum_{\theta\in\Theta}\sum_{s\in\mathcal S_q}
\mu_0(\theta)\widehat q_N(s\mid\theta)u(h(s),\theta)
=
\sum_{s\in\mathcal S_q}\widehat W_{h(s)}(s).
\]
Only signals appearing in the finite training sample have a nonzero summand.  The choice of $h(s)$ therefore separates signal by signal, and $\widehat\alpha_N(s)$ maximizes each summand.  Hence it maximizes their sum.
\end{proof}

Let
\[
r_c:=\frac{(1-c)V^{\mathrm{NI}}}{2}.
\]
Fix $h\in\mathcal H_q$.  The variables
\[
X^h_{\theta,i}
:=
\frac{\mu_0(\theta)}{N}
u\bigl(h(S_{\theta,i}),\theta\bigr),
\qquad
\theta\in\Theta,\quad i\in\{1,\ldots,N\},
\]
are independent, satisfy
\[
0<X^h_{\theta,i}
\leq\frac{\mu_0(\theta)U}{N},
\]
and obey
\[
\sum_{\theta,i}X^h_{\theta,i}
=\widehat{\operatorname{VAL}}_N(h),
\qquad
\mathbb E\!\left[\sum_{\theta,i}X^h_{\theta,i}\right]
=\operatorname{VAL}_h(q).
\]
Moreover, the sum of the squared interval lengths is
\[
\sum_{\theta\in\Theta}\sum_{i=1}^N
\left(\frac{\mu_0(\theta)U}{N}\right)^2
=
\frac{U^2C_\mu}{N}.
\]
By Hoeffding's inequality,
\[
\Pr\!\left[
\left|
\widehat{\operatorname{VAL}}_N(h)-\operatorname{VAL}_h(q)
\right|
\geq t
\right]
\leq
2\exp\!\left(-\frac{2Nt^2}{U^2C_\mu}\right).
\]
Since $|\mathcal H_q|=|A|^{K_q}$, substituting $t=r_c$ and taking a union bound yields
\[
\Pr\!\left[
\sup_{h\in\mathcal H_q}
\left|
\widehat{\operatorname{VAL}}_N(h)-\operatorname{VAL}_h(q)
\right|
>r_c
\right]
\leq
2\exp\!\left(
K_q\log|A|
-\frac{N(1-c)^2(V^{\mathrm{NI}})^2}{2U^2C_\mu}
\right)
\leq\delta.
\]
On the complementary event, let $\alpha^*$ be a Bayes-optimal policy, so that $\operatorname{VAL}_{\alpha^*}(q)=\operatorname{OPT}(q)$.  By \autoref{lem:empirical-bayes-erm},
\[
\frac{\operatorname{VAL}_{\widehat\alpha_N}(q)}{\operatorname{OPT}(q)}
\geq
\frac{\widehat{\operatorname{VAL}}_N(\widehat\alpha_N)-r_c}
{\operatorname{OPT}(q)}
\geq
\frac{\widehat{\operatorname{VAL}}_N(\alpha^*)-r_c}
{\operatorname{OPT}(q)}
\geq
1-\frac{2r_c}{\operatorname{OPT}(q)}
\geq
1-\frac{2r_c}{V^{\mathrm{NI}}}
=c.
\]
\end{proof}

\section{Sample Lower Bounds}
\label{appendix:monte-carlo-lower-bounds}

Consider the fixed two-state matching environment
\begin{equation}
\label{eq:hidden-sign-environment}
\Theta=A=\{-1,+1\},
\qquad
\mu_0(-1)=\mu_0(+1)=\frac12,
\qquad
u(a,\theta)=\frac{1+\mathbf 1\{a=\theta\}}{2}.
\end{equation}
Let $K=2m$ be even, and let the signal set be
\[
\mathcal S_K
:=
\{(j,x):j\in\{1,\ldots,m\},\ x\in\{-1,+1\}\}.
\]
For a hidden vector $b=(b_1,\ldots,b_m)\in\{-1,+1\}^m$ and $\gamma\in(0,1/2]$, define
\[
q_b((j,x)\mid\theta)
:=
\frac{1+\theta\gamma b_jx}{K}.
\]
This is a valid experiment.  Indeed, all probabilities are positive and, for every $\theta$,
\[
\sum_{j=1}^m\sum_{x\in\{-1,+1\}}
q_b((j,x)\mid\theta)
=
\frac{2m}{K}=1.
\]
The marginal probability of each signal is
\[
Z(q_b,(j,x))
=
\frac12\frac{1+\gamma b_jx}{K}
+
\frac12\frac{1-\gamma b_jx}{K}
=
\frac1K.
\]
By Bayes' rule,
\[
\Pr(\theta=t\mid j,x)
=
\frac{1+t\gamma b_jx}{2},
\qquad t\in\{-1,+1\},
\]
and hence
\[
\mathbb E[\theta\mid j,x]
=
\gamma b_jx.
\]
Thus the posterior-optimal action is $b_jx$.  Its posterior payoff is $(3+\gamma)/4$, whereas the payoff of the other action is $(3-\gamma)/4$.  In particular,
\[
\operatorname{OPT}(q_b)=\frac{3+\gamma}{4}.
\]
For every deterministic policy $h$, define its fraction of incorrect signalwise choices by
\[
e_b(h)
:=
\frac{1}{K}
\sum_{j=1}^m\sum_{x\in\{-1,+1\}}
\mathbf 1\{h(j,x)\neq b_jx\}.
\]
Then
\begin{equation}
\label{eq:hidden-sign-value-ratio}
\frac{\operatorname{VAL}_h(q_b)}{\operatorname{OPT}(q_b)}
= \frac{\frac{3+\gamma}{4}-\frac{\gamma}{2}e_b(h)}{\frac{3+\gamma}{4}} =
1-\frac{2\gamma}{3+\gamma}e_b(h).
\end{equation}

A \emph{fixed-training-sample forward-sampling selector} with sample size $N$ is a randomized learning rule that observes, for each state $\theta$, $N$ independent draws from $q(\cdot\mid\theta)$ and has no other access to the experiment.  From this training transcript, including any auxiliary randomness, it produces a deterministic policy $\widehat\alpha$.  The transcript is held fixed throughout evaluation: on receiving a query signal, the policy makes no additional calls to the sampler.

\begin{proposition}[Monte Carlo $(c,\delta)$-PAC sample lower bound]
\label{prop:monte-carlo-pac-lower}
There is a universal constant $C_0>0$ with the following property.  For every even $K\geq2$, every $c\in[31/32,1)$, and every $\delta\in(0,1/64]$, set $\gamma:=16(1-c)$.  Any fixed-training-sample forward-sampling selector that uses $N$ independent samples from each state in \eqref{eq:hidden-sign-environment} and satisfies the $(c,\delta)$-PAC guarantee
\[
\Pr\!\left[
\operatorname{VAL}_{\widehat\alpha}(q_b)
\geq
c\,\operatorname{OPT}(q_b)
\right]
\geq1-\delta
\]
uniformly over $b\in\{-1,+1\}^{K/2}$ must have
\begin{equation}
\label{eq:batch-sample-lower}
N
\geq
C_0
\frac{K+\log(1/\delta)}{(1-c)^2}.
\end{equation}
\end{proposition}

\begin{proof}[Proof of \autoref{prop:monte-carlo-pac-lower}]
Since $c\geq31/32$, we have $\gamma=16(1-c)\leq1/2$.  We prove the $K$ and confidence terms separately.

\noindent\textbf{Step 1: Bound the divergence between neighboring hidden vectors.}

Let $b^{(j)}$ be obtained from $b$ by flipping only coordinate $j$.  For either fixed state $\theta$, the two conditional signal distributions differ only on $(j,-1)$ and $(j,+1)$. We can directly calculate the KL divergence,
\begin{equation}
\label{eq:hidden-sign-one-sample-kl}
D_{\mathrm{KL}}\!\left(
q_b(\cdot\mid\theta)
\,\middle\|\,
q_{b^{(j)}}(\cdot\mid\theta)
\right)
=
\frac{2\gamma}{K}
\log\frac{1+\gamma}{1-\gamma}
\leq
\frac{8\gamma^2}{K}.
\end{equation}
The inequality uses $\log((1+\gamma)/(1-\gamma))\leq4\gamma$ for $\gamma\leq1/2$.  For hidden vector $b$, let the complete training transcript be
\[
T_b
:=
\left(
R,
(S_{\theta,i})_{\theta\in\{-1,+1\},\,1\leq i\leq N}
\right),
\]
where $S_{\theta,1},\ldots,S_{\theta,N}$ are independent draws from $q_b(\cdot\mid\theta)$ and $R$ contains any auxiliary randomness used to form the fixed policy.  The random variable $R$ is independent of the training observations and has the same distribution under every hidden vector.  Let $\mathbb P_b$ denote the distribution of $T_b$.  The resulting fixed policy is a deterministic function of $T_b$.  Since $R$ contributes zero relative entropy, additivity over the independent observations and \eqref{eq:hidden-sign-one-sample-kl} give
\begin{equation}
\label{eq:hidden-sign-neighbor-transcript-kl}
\begin{aligned}
D_{\mathrm{KL}}(\mathbb P_b\|\mathbb P_{b^{(j)}})
&=
\sum_{\theta\in\{-1,+1\}}
N D_{\mathrm{KL}}\!\left(
q_b(\cdot\mid\theta)
\,\middle\|\,
q_{b^{(j)}}(\cdot\mid\theta)
\right)\\
&\leq
2N\frac{8\gamma^2}{K}
=
\frac{16N\gamma^2}{K}.
\end{aligned}
\end{equation}

\noindent\textbf{Step 2: Obtain the bound for $K$.}

From the induced deterministic selector define an estimator of the hidden vector by
\[
\widehat b_j:=\widehat\alpha(j,+1),
\qquad j=1,\ldots,m.
\]
Every coordinate error $\widehat b_j\neq b_j$ makes the policy choose the wrong action at signal $(j,+1)$.  Therefore
\begin{equation}
\label{eq:error-dominates-hamming}
e_b(\widehat\alpha)
\geq
\frac{1}{K}d_{\mathrm H}(\widehat b,b)
\end{equation}
where $d_{\mathrm H}(\widehat b,b)$ is the Hamming distance between $\widehat b$ and $b$. 

Let $b$ be uniformly distributed over the hypercube of hidden vectors.  Pairing $b$ with $b^{(j)}$ and using the elementary testing bound that the sum of the two errors is at least one minus total variation yields
\[
\mathbb E_b\mathbb E_b^{\mathbb P}
\bigl[d_{\mathrm H}(\widehat b,b)\bigr]
\geq
\frac{m}{2}
\left(
1-\max_{b,j}d_{\mathrm{TV}}(\mathbb P_b,\mathbb P_{b^{(j)}})
\right).
\]
By Pinsker's inequality and \eqref{eq:hidden-sign-neighbor-transcript-kl}, 
\[
d_{\mathrm{TV}}(\mathbb P_b,\mathbb P_{b^{(j)}})
\leq
\sqrt{\frac{8N\gamma^2}{K}}.
\]
If $N\leq K/(128\gamma^2)$, this distance is at most $1/4$.  Since $m=K/2$, \eqref{eq:error-dominates-hamming} and the preceding display then imply
\begin{equation}
\label{eq:assouad-error-lower}
\mathbb E_b\mathbb E_b^{\mathbb P}
\bigl[e_b(\widehat\alpha)\bigr]
\geq
\frac{3}{16}.
\end{equation}

On the other hand, \eqref{eq:hidden-sign-value-ratio} shows that whenever the $(c,\delta)$-PAC guarantee succeeds,
\[
e_b(\widehat\alpha)
\leq
\frac{(1-c)(3+\gamma)}{2\gamma}
=
\frac{3+\gamma}{32}
\leq
\frac{7}{64}.
\]
Since $0\leq e_b(\widehat\alpha)\leq1$ and $\delta\leq1/64$, it follows for every $b$ that
\[
\mathbb E_b^{\mathbb P}
\bigl[e_b(\widehat\alpha)\bigr]
\leq
\frac{7}{64}+\delta
\leq
\frac{1}{8}.
\]
This contradicts \eqref{eq:assouad-error-lower}.  Hence
\begin{equation}
\label{eq:batch-support-lower}
N
>
\frac{K}{128\gamma^2}
=
\Omega\!\left(\frac{K}{(1-c)^2}\right).
\end{equation}

\noindent\textbf{Step 3: Obtain the confidence term.}

Fix any $b$ and consider the binary testing problem
\[
H_b:T\sim\mathbb P_b
\qquad\text{versus}\qquad
H_{-b}:T\sim\mathbb P_{-b}.
\]
The posterior-optimal action at signal $(j,x)$ is $b_jx$ under $q_b$ and $-b_jx$ under $q_{-b}$.  Since the action set is $\{-1,+1\}$, every policy $h$ therefore satisfies, signal by signal,
\[
\mathbf 1\{h(j,x)\neq b_jx\}
+
\mathbf 1\{h(j,x)\neq-b_jx\}
=1.
\]
Averaging this identity over the $K$ signals gives
\begin{equation}
\label{eq:antipodal-error-sum}
e_b(h)+e_{-b}(h)=1.
\end{equation}

Given a realized transcript $T$, let $h_T$ be the fixed policy produced by the training rule.  Define the test
\[
\phi(T)
:=
\begin{cases}
b, & \text{if }e_b(h_T)\leq e_{-b}(h_T),\\
-b, & \text{otherwise}.
\end{cases}
\]
This is a valid information-theoretic test because the two candidate vectors are fixed and $h_T$ is a total policy on the finite signal set.  By \eqref{eq:hidden-sign-value-ratio}, a policy satisfying the multiplicative guarantee under $q_b$ obeys
\[
e_b(h_T)
\leq
\frac{(1-c)(3+\gamma)}{2\gamma}
=
\frac{3+\gamma}{32}
\leq
\frac{7}{64}
<
\frac12.
\]
Equation~\eqref{eq:antipodal-error-sum} then implies $e_{-b}(h_T)>1/2$, so the test selects $b$.  Hence the uniform $(c,\delta)$-PAC guarantee gives
\[
\mathbb P_b\!\left[\phi(T)=-b\right]
\leq
\mathbb P_b\!\left[
\operatorname{VAL}_{h_T}(q_b)
<c\,\operatorname{OPT}(q_b)
\right]
\leq\delta.
\]
The same argument with $b$ and $-b$ interchanged gives
\[
\mathbb P_{-b}\!\left[\phi(T)=b\right]
\leq\delta.
\]

We next bound the divergence between the two transcript distributions.  For either fixed state $\theta$, the two signal probabilities are interchanged within every pair $(j,-1),(j,+1)$.  The contribution of each pair to relative entropy is the expression in \eqref{eq:hidden-sign-one-sample-kl}; summing over the $m=K/2$ coordinates gives
\[
\begin{aligned}
D_{\mathrm{KL}}\!\left(
q_b(\cdot\mid\theta)
\,\middle\|\,
q_{-b}(\cdot\mid\theta)
\right)
&=
\sum_{j=1}^m
\frac{2\gamma}{K}
\log\frac{1+\gamma}{1-\gamma}\\
&=
\gamma\log\frac{1+\gamma}{1-\gamma}.
\end{aligned}
\]
There are $N$ independent observations from each of the two states, while the auxiliary randomness again contributes zero relative entropy.  Consequently,
\[
\begin{aligned}
D_{\mathrm{KL}}(\mathbb P_b\|\mathbb P_{-b})
&=
\sum_{\theta\in\{-1,+1\}}
N D_{\mathrm{KL}}\!\left(
q_b(\cdot\mid\theta)
\,\middle\|\,
q_{-b}(\cdot\mid\theta)
\right)\\
&=
2N\gamma\log\frac{1+\gamma}{1-\gamma}
\leq
8N\gamma^2,
\end{aligned}
\]
where the inequality again uses $\log((1+\gamma)/(1-\gamma))\leq4\gamma$.  The Bretagnolle--Huber testing inequality \citep{BretagnolleHuber1979}, applied to the test $\phi$, now gives
\[
2\delta
\geq
\mathbb P_b\!\left[\phi(T)=-b\right]
+
\mathbb P_{-b}\!\left[\phi(T)=b\right]
\geq
\frac12
\exp\!\left(-D_{\mathrm{KL}}(\mathbb P_b\|\mathbb P_{-b})\right)
\geq
\frac12e^{-8N\gamma^2}.
\]
Equivalently, $e^{-8N\gamma^2}\leq4\delta$, and hence
\begin{equation}
\label{eq:batch-confidence-lower}
N
\geq
\frac{1}{8\gamma^2}
\log\frac{1}{4\delta}
=
\frac{1}{2048(1-c)^2}
\log\frac{1}{4\delta}
=
\Omega\!\left(
\frac{\log(1/\delta)}{(1-c)^2}
\right).
\end{equation}
The final comparison uses $\gamma=16(1-c)$ and $\delta\leq1/64$, which implies $\log(1/(4\delta))\geq(2/3)\log(1/\delta)$.
Taking the maximum of \eqref{eq:batch-support-lower} and \eqref{eq:batch-confidence-lower}, and using $\max\{x,y\}\geq(x+y)/2$, proves \eqref{eq:batch-sample-lower} after adjusting the universal constant $C_0$.
\end{proof}

\section{Preliminaries}
\label{appendix:computational-preliminaries}

This section gives an informal, self-contained guide to the computational concepts used in the paper.

\subsection{Circuits and representations}
\label{subsec:circuits-and-representations}

A \emph{Boolean circuit} is a finite directed acyclic network of elementary logical operations. Its input wires carry bits, each internal node (or \emph{gate}) applies a fixed Boolean operation such as \textsc{And}, \textsc{Or}, or \textsc{Not} to the bits on its incoming wires, and designated output wires return one or more bits. Acyclicity allows the gates to be evaluated in an order in which every input to a gate has already been computed. Other standard finite, functionally complete collections of gates are equivalent for our purposes, because one such gate basis can be translated into another with at most polynomial overhead.

A \emph{randomized Boolean circuit} has, in addition to its ordinary input wires, random input wires containing independent fair bits. Conditional on the ordinary input and on a realization of these random bits, the circuit is deterministic. In our model the ordinary input is a fixed binary encoding of the state $\theta$, the random inputs form the seed $R$, and the output bits encode the signal $s=q(\theta;R)$. If the circuit uses $r_q$ random bits, then
\[
q(s\mid\theta)
=
2^{-r_q}\bigl|\{R\in\{0,1\}^{r_q}:q(\theta;R)=s\}\bigr|.
\]
Thus, one forward draw is easy to generate by choosing $R$ and evaluating the circuit, even though computing the likelihood of an observed signal may require aggregating over exponentially many seeds. This difference between forward simulation and likelihood evaluation is central to the paper.

\paragraph{Gate-list encoding.}
Every acyclic circuit can be numbered in topological order. A standard gate-list description records the numbers and roles of the ordinary inputs, random inputs, and outputs and, for each gate, its type and the indices of the earlier wires that feed it. The resulting binary string is denoted by $\langle q\rangle$. With a bounded-fan-in gate basis, each predecessor index requires only logarithmically many bits, so the gate-list length is polynomially related to the number of gates and wires. The represented circuit can therefore be evaluated in time polynomial in $|\langle q\rangle|$. The precise reasonable choice of gate basis, numbering convention, or gate-list format does not affect any polynomial-time classification, since standard choices can be translated into one another in polynomial time.

\paragraph{Generality of the representation.}
Circuit representation is not tied to a particular economic or statistical mechanism. States, records, messages, and actions can all be encoded as bit strings, and circuits can compose arbitrary finite logical and arithmetic subroutines on those encodings. More generally, at any fixed input length, a deterministic machine or algorithm that runs for at most a bounded number of steps can be unrolled into a Boolean circuit whose size is polynomial in that running time; a randomized algorithm is represented in the same way by exposing its random tape as additional input wires. Conversely, an explicitly given circuit can be evaluated efficiently from its gate list. Hence, circuits represent general finite, bounded-time processes, machines, and algorithms up to polynomial overhead, which is the level of generality relevant to the complexity classes used below. 

\subsection{Complexity classes}
\label{subsec:complexity-classes}

Complexity classes compare how the resources needed to solve a computational yes-or-no problem grow with the bit length $n$ of its input. Here \emph{polynomial time} means that the number of elementary computational steps is bounded by $Cn^k$ for some constants $C$ and $k$ independent of the instance. 

\paragraph{Deterministic and verifiable computation.}
The class $\mathsf{P}$ consists of problems for which a deterministic algorithm can decide the correct answer in polynomial time. The class $\mathsf{NP}$ consists of problems for which every yes-instance has a certificate of polynomial length whose validity can be checked deterministically in polynomial time. Boolean satisfiability (\textsc{SAT})--deciding whether a Boolean formula has at least one satisfying assignment--is a canonical $\mathsf{NP}$-complete problem.

\paragraph{Randomized computation.}
The class $\mathsf{RP}$ consists of problems admitting a randomized polynomial-time algorithm with one-sided error: it never accepts a no-instance, and it accepts every yes-instance with probability at least $1/2$. The particular fixed positive constant can be changed by independent repetition. The class $\mathsf{BPP}$ allows two-sided error: a randomized polynomial-time algorithm gives the correct answer with probability at least $2/3$ on every instance. Again, any fixed success probability strictly above $1/2$ gives the same class because repetition can amplify the success probability.

The class $\mathsf{PP}$ also uses randomized polynomial time, but it requires only a strict majority of accepting random tapes: a yes-instance is accepted with probability greater than $1/2$, whereas a no-instance is accepted with probability at most $1/2$. In \textsc{Majsat}, one asks whether a Boolean formula is satisfied by strictly more than half of all assignments; this is a canonical $\mathsf{PP}$-complete problem \citep{Gill1977}.

\paragraph{Statistical zero knowledge.}
The class $\mathsf{SZK}$ consists of problems having interactive proofs in which a computationally unbounded prover can convince a randomized polynomial-time verifier of a yes-instance, a cheating prover cannot convincingly establish a no-instance except with bounded probability, and the interaction reveals essentially no additional information: on yes-instances, the verifier's view can be efficiently simulated without the prover up to negligible statistical distance. A canonical $\mathsf{SZK}$-complete problem is Statistical Difference \citep{SahaiVadhan2003}. Given descriptions of two sampling circuits, it asks whether their output distributions are far apart or close in total variation distance.

\paragraph{Containments and standard conjectures.}
The unconditional containments relevant to this paper include
\[
\mathsf{P}\subseteq\mathsf{RP}\subseteq\mathsf{BPP}\subseteq\mathsf{PP},
\qquad
\mathsf{RP}\subseteq\mathsf{NP}\subseteq\mathsf{PP},
\qquad
\mathsf{BPP}\subseteq\mathsf{SZK}.
\]
It is widely conjectured that $\mathsf{P}\neq\mathsf{NP}$. Since $\mathsf{P}\subseteq\mathsf{NP}\subseteq\mathsf{PP}$, that conjecture implies $\mathsf{P}\neq\mathsf{PP}$, so the latter is a logically weaker separation assumption. It is also widely conjectured that randomization does not enlarge polynomial-time computation, namely $\mathsf{P}=\mathsf{BPP}$; this would imply $\mathsf{RP}=\mathsf{P}$ and, together with $\mathsf{P}\neq\mathsf{NP}$, would imply $\mathsf{RP}\neq\mathsf{NP}$. Finally, the paper assumes $\mathsf{BPP}\neq\mathsf{SZK}$ for its ex ante lower bounds. Because $\mathsf{BPP}\subseteq\mathsf{SZK}$, this assumption says that the containment is strict; as noted in the main text, it is a stronger assumption than $\mathsf{P}\neq\mathsf{NP}$ \citep{SahaiVadhan2003}. 
\end{document}